\documentclass{article}

\usepackage[doublespacing]{setspace}
\usepackage{supertech}

\usepackage{fullpage}
\usepackage{hyperref}
\usepackage{xspace}
\usepackage{amsmath}
\usepackage{amssymb}
\usepackage{cite}
\usepackage{psfig}
\usepackage{epsfig}
\usepackage{psfrag}       
\usepackage{subcaption}

\usepackage[linesnumbered,ruled,vlined]{algorithm2e}

\usepackage{multirow}
\usepackage{multicol}

\newcommand{\parhead}[1]{{\textbf{#1.}\xspace}}

\newcommand{\ADN}{{Anonymous Dynamic Network}\xspace}
\newcommand{\ADNs}{{Anonymous Dynamic Networks}\xspace}

\newcommand{\MC}{\textsc{Methodical Counting}\xspace}

\newcommand{\LLMC}{\textsc{Leaderless \MC}\xspace}
\newcommand{\nameD}{MMC\xspace}
\newcommand{\nameR}{LLMC\xspace}
\newcommand{\nameT}{MMCT\xspace}
\newcommand{\ellp}{\ell'\xspace}
\newcommand{\ldr}{{black}\xspace}

\newcommand{\notldr}{{white}\xspace}

\newcommand{\cA}{{\mathcal{A}}\xspace}
\newcommand{\gadget}{{gadget}\xspace}
\newcommand{\gadgets}{{gadgets}\xspace}

\usepackage{xcolor}
\newcommand{\miguel}[1]{\textcolor{blue}{#1}}

\renewcommand{\miguel}[1]{#1}
\usepackage[normalem]{ulem} 

\newtheorem{proposition}{Proposition}
\newtheorem{definition}{Definition}
\newtheorem{claim}{Claim}
\newcommand{\remove}[1]{}

\begin{document}

\title{Polynomial Anonymous Dynamic Distributed Computing\\ without a Unique Leader\thanks{An extended abstract of this work has been presented in~\cite{icalp19counting}.}}

\author{
Dariusz~R.~Kowalski~\thanks{School of Computer and Cyber Sciences, Augusta University, GA, USA; SWPS University of Social Sciences and Humanities, Warsaw, Poland; \href{mailto:dkowalski@augusta.edu}{dkowalski@augusta.edu}}
\and
Miguel~A.~Mosteiro~\thanks{Computer Science Department, Pace University, New York, NY, USA; \href{mailto:mmosteiro@pace.edu}{mmosteiro@pace.edu}}
}



\date{}

\maketitle

\begin{abstract}

Counting the number of nodes in \ADNs is enticing from an algorithmic perspective: an important computation in a restricted platform with promising applications.
Starting with Michail, Chatzigiannakis, and Spirakis~\cite{spirakis}, a flurry of papers sped up the running time guarantees from doubly-exponential to polynomial~\cite{KowalskiMicalp18}. There is a common theme across all those works: a distinguished node is assumed to be present, because Counting cannot be solved deterministically without at least one. 

In the present work we study challenging questions that naturally follow: how to 
efficiently count with more than one distinguished node, or how to count without any distinguished node. More importantly, what is the minimal information needed about these distinguished nodes and what is the best we can aim for (count precision, stochastic guarantees, etc.) without any. 
We present negative and positive results to answer these questions. To the best of our knowledge, this is the first work that addresses them.

\end{abstract}


\section{Introduction}\label{sec:introduction}

In the \ADN (ADN) model, network nodes lack any form of identifiers, that is, they cannot be distinguished among themselves, and communication links between nodes may change arbitrarily and continuously over time.

The recent excitement around the problem of \emph{Counting}, i.e., computing the number of nodes of an \ADNs comes as no surprise. On one hand, knowing the number of computational entities is a fundamental requirement to decide termination in a myriad of distributed algorithms. 
On the other hand, ADN is a very restrictive scenario from an algorithmic perspective. 
The combination of an important problem with harsh computational conditions is any algorithmist's delight. 

Node anonymity in ADNs is motivated by expected applications of such communication infrastructure. For instance, in ad-hoc networks embedded in the Internet of Things, nodes may have to be deployed in a massive scale, and having unique identifiers may simply be impractical or inconvenient. Moreover, low node-cost expectations may introduce uncertainty about the number of nodes that will effectively startup. Hence, the need of Counting. 

Strikingly, the progress on deterministic Counting speed-ups over various works ranged over a broad spectrum: from unbounded~\cite{conscious, oracle} to polynomial time~\cite{KowalskiMicalp18}, going through doubly-exponential~\cite{conscious} and exponential time~\cite{LunaB15,ChakrabortyMM18}. In all that fruitful work, the ADN model has been equipped with one distinguishable node\footnote{Exactly one, usually called \emph{leader}. We refrain from using the name leader to avoid confusion: in our model we may have more than one, and in our algorithm they are not going to select a single one among them. Moreover, they cannot even be local leaders, because due to ADN's dynamicity they may all be connected being their own (local) leaders.}. The assumption is well motivated: in a seminal paper~\cite{spirakis}, it was shown that, without at least one such node, Counting cannot be solved deterministically. 
But what if we have more than one special node? 
In what sense these nodes need to be special?
Is it enough to put one of two different programs on each of the nodes, that otherwise are all identical?
Moreover, can we let the nodes choose at random which program to run and have no special nodes?
To the best of our knowledge, this is the first work that considers these questions.

As the more general case where all nodes are identical, only differentiated by the program they run, let the set of $n$ network nodes be formed by $\ell$ \emph{\ldr} nodes (the ``special'' ones) and $n-\ell$ \emph{\notldr} nodes (the ``regular'' ones).   
 {\bf Our first contribution} is negative results. On one hand, if $\ell$ is unknown, Counting cannot be solved deterministically. 
On the other hand, we show that 
there exist ADNs such that, if $\ell = 0$ or $\ell$ is unknown to the nodes, there is no Counting randomized algorithm such that, with constant probability, each node outputs the correct count when it stops.

Even knowing how many \ldr nodes are in the network, straightforward application of previous ideas for Counting is not clear. Indeed, each \ldr node may carry its own count, but how do we combine or compare final counts? Even passing messages among \ldr nodes is challenging, because \ldr nodes are also indistinguishable among them and communication links change arbitrarily. For instance, a \ldr node is not able to tell whether a received message is even its own, coming back after being previously sent for dissemination. 

{\bf Our second contribution} is a deterministic Counting protocol that computes exactly the number of network nodes. 
Our protocol uses no information about the network, except the number of \ldr nodes $\ell\geq 1$. 
After completing its execution, all nodes obtain the exact size of the network and stop.
Moreover, they stop all at the same time, allowing the algorithm to be concatenated with other computations.

This protocol resembles our \MC protocol presented in~\cite{KowalskiMicalp18}. So, we call it \textsc{Methodical multi-Counting} (\nameD). However, it is not a simple combination of multiple instances of \MC to handle multiple \ldr nodes. We overcome the challenge of how to combine the actions of multiple indistinguishable \ldr nodes by careful design of a set of alarms, so that \emph{all} \ldr nodes can simultaneously detect when a running estimate of the size is correct. 
Moreover, the asymptotic performance of \nameD is $O(n^{4+\epsilon}(\log^3 n)/\ell)$, for an arbitrarily small $\epsilon>0$. This is a speed-up by a factor arbitrarily close to $n\ell/\log n$ with respect to \MC. That is, even for $\ell=1$, \nameD is faster than the best previous work (cf. Table~\ref{table}). 

In face of the impossibility of deterministic Counting without a distinguished node~\cite{spirakis}, a natural question is what is possible introducing randomness. Our impossibility results show that if $\ell=0$ or $\ell$ is unknown it is not possible to solve Counting with a constant probability. Thus, we focus on a weaker version of the problem called \emph{Unconscious Counting}~\cite{conscious}. In Unconscious Counting nodes still compute the network size, but may not know whether the number they have at the moment is precise 
or will be corrected later (by additional communication received).

{\bf Our third contribution} is an Unconscious Counting protocol that computes exactly the number of nodes in an ADN where $\ell=0$ (no special nodes). That is, 
for the first time it is possible to consider an ADN model where \emph{all} nodes are identical, indistinguishable, and run the same program.
Our protocol, called \textsc{\LLMC} (\nameR), is Monte Carlo. That is, there is a small probability at most $\zeta>0$ of obtaining the wrong size, and the time of $O(\eta^{4+\epsilon}\log^3 \eta)$, for $\eta=\max\{n,\lceil\lceil12/\epsilon\rceil\rceil\}$ and $\epsilon>0$, to obtain the correct network size holds with probability at least $1-\zeta$ (cf. Table~\ref{table}).

To the best of our knowledge, this is the first comprehensive study of Counting in ADNs where $\ell$ may be different than one. 

\parhead{Roadmap} 
The rest of the paper is organized as follows. We specify the model and notation details in Section~\ref{prelim}. 
The directly related work overviewed in comparison with our results is summarized in Table~\ref{table}.
Further details on previous work 
are included in Section~\ref{relwork}.
We present impossibility results in Section~\ref{sec:impossibility}. 
Sections~\ref{algorithm} and~\ref{sec:randomized} include the details of \nameD and \nameR. 
References to algorithms lines are given as $\langle algorithm\#\rangle.\langle line\#\rangle$.
We conclude this study with some open problems in Section~\ref{sec:conclude}.

\section{Model, Problem, and Notation}
\label{prelim}

\subsubsection*{\ADNs.}
The following model is customary in the \ADNs literature.

We consider a network composed by a set $V$ of $n>1$ network \emph{nodes} with processing and communication capabilities. 
Each pair of nodes that are able to communicate defines a communication \emph{link}, and the set of links is called the \emph{topology} of the network. The nodes in a communication link are called \emph{neighbors}.

Without loss of generality, we discretize time in \emph{communication rounds} or \emph{time slots} indistinctively (or \emph{rounds} or \emph{slots} for short).
In any given round, a node may send a \emph{message} through all its communication links, receive all messages from all sending neighbors, and carry out some (local) computations, in that order. 
The time needed for computations is assumed negligible with respect to the time needed for communication, and the size of messages is not bounded.

The set of links among nodes may change from round to round\footnote{This dynamicity is the reason for the term "Dynamic" in \ADN.}.
These topology changes are arbitrary, limited only to maintain 
the connectedness of the network in each round. 
That is, at any given round the topology is such that there is a \emph{path}, i.e., a sequence of links, between each pair of nodes, but the set of links may change arbitrarily from round to round. This adversarial model of dynamics is known as \emph{$1$-interval connectivity} in~\cite{KuhnLO2010}.

The event of sending a message to neighbors is called a \emph{broadcast} or \emph{transmission}. 
Nodes and links are reliable, in the sense that no 
link or node failures occur.
Hence, a broadcasted message is received by all current neighbors.
Moreover, links are \emph{symmetric}, that is, if node $a$ is able to send a message to node $b$, then $b$ is able to send a message to $a$ in the same round.

Nodes in the ADN model lack any form of identifiers\footnote{The lack of identifiers is the reason for the term "Anonymous" in the model's name ``\ADN''.}. That is, they are indistinguishable from any other node. 
However, it was shown in~\cite{spirakis} that the Counting Problem (defined below) cannot be solved 
deterministically
in Anonymous Networks without the availability of at least one distinguished node in the network.
Hence, all previous studies of Counting in ADNs included in the model the presence of such node called the \emph{leader}.
However, to the best of our knowledge, nothing is known about deterministic Counting in presence of \emph{multiple} distinguished nodes.
In this work, we generalize the ADN model assuming that the number of distinguished nodes is $\ell\geq 1$.
Aside from the distinction between distinguished and not-distinguished, all nodes are indistinguishable within their group. Therefore, here we call them \ldr nodes and \notldr nodes, respectively. All \ldr nodes execute exactly the same program, and all \notldr nodes execute exactly the same program. That is, there are no identifiers that allow to distinguish one \ldr (resp. \notldr) node from another \ldr (resp. \notldr) node\footnote{Nodes are labeled throughout the paper only for the sake of presentation and analysis.}.


\subsubsection*{The Counting Problem.}

We define the Counting problem as follows.
\begin{definition}
An algorithm $\mathcal{A}$ solves the \emph{Counting Problem} if, after completing its execution, all network nodes running $\mathcal{A}$ have obtained the size of the network and stop (not necessarily concurrently).
\end{definition}
Following up on~\cite{conscious}, we also define a weaker version of Counting, as follows.
\begin{definition}
An algorithm $\mathcal{A}$ solves the \emph{Unconscious Counting Problem} if all network nodes running $\mathcal{A}$ eventually obtain the size of the network in finite time.
\end{definition}
Notice that the definition of Unconscious Counting does not require the nodes to know that they have obtained the correct count by the current round. Consequently, Unconscious Counting algorithms may not satisfy a termination condition. 

Notice also that we focus on \emph{exact} Counting. That is, all nodes obtain the exact size of the network $n$, rather than an approximation as other works.

We define now a class of Counting algorithms.
\begin{definition}
For a given algorithm $\mathcal{A}$, let an \emph{execution} of $\mathcal{A}$ be a sequence of steps of $\mathcal{A}$ followed in one of the possible sequence of choices made by $\mathcal{A}$. Let $\mathcal{X}(\mathcal{A})$ be the set of all possible executions of $\mathcal{A}$. 
An algorithm $\mathcal{A}$ is called \emph{eventually stopping} if, for all $X\in \mathcal{X}(\mathcal{A})$, $X$ has finite length.
\end{definition}
We will model worst case scenarios assuming the presence of an adversary that controls the topology of the network. In particular, we consider the following adversaries. 
\begin{definition}
Let the sequence $\mathcal{E}=\langle E_1,E_2,\dots\rangle$ be the sets of communication links of an ADN for time slots $t_1,t_2,\dots$.
Consider the execution of an algorithm $\mathcal{A}$.
We say that an adversary is \emph{oblivious} if it determines the sequence $\mathcal{E}$ completely before the execution of $\mathcal{A}$ begins. 
On the other hand, we say that an adversary is \emph{adaptive} if it determines the sequence $\mathcal{E}$ during the execution of $\mathcal{A}$, according to the actions of $\mathcal{A}$.
\end{definition}
Notice that the distinction between oblivious and adaptive makes sense only for randomized algorithms, given that for deterministic algorithms the actions of the algorithm are defined before the execution.


\subsubsection*{Notation.}
The following notation will be used. 
In this work, the maximum number of neighbors that any node may have at any given time, called the \emph{dynamic maximum degree}, 
is denoted as $d_{\max}$.\footnote{Previous literature on ADNs denotes the dynamic maximum degree as $\Delta$, which is classically used in static networks to denote the maximum degree. We have chosen $d_{\max}$ to avoid confusion with static networks notation. This choice is additionally consistent with previous work on lazy random walks, where the degree is denoted by $d$.} 
The maximum length of a path between any pair of nodes at any given time is called the \emph{dynamic diameter} and it is denoted as $D$. 
The maximum length of an opportunistic path between any pair of nodes over many time slots (i.e., a path that may not exist in its entirety at any given time, but that it can be followed by a message along time as new links appear while links already traversed may disappear) is called the \emph{chronopath}~\cite{FCFMMZ:randomgeocast} and it is denoted as $\mathcal{D}$.

\begin{table}[htbp]
\centering
\caption{Statement of results and comparison with previous work.}
\label{table}
\resizebox{\textwidth}{!}{
\begin{tabular}{|c|c|c|c|c|c|c|c|}
\hline
\rule{0pt}{4ex}
\multirow{2}{*}{algorithm}&\multicolumn{3}{c|}{needs}&\multirow{2}{*}{computes}&\multirow{2}{*}{stops?}&\multicolumn{2}{c|}{complexity}\\
[.1in]
\cline{2-4}
\cline{7-8}
\rule{0pt}{4ex}
&\begin{tabular}{c}distinguished\\nodes\end{tabular}&\begin{tabular}{c}size\\ upper\\ bound\\ $N$\end{tabular}&\begin{tabular}{c}dynamic\\ maximum\\ degree u.b.\\$d_{\max}$\end{tabular}&&&time&space\\
[.1in]
\hline
\hline
\rule{0pt}{4ex}
\begin{tabular}{c}\emph{Degree}\\\emph{Counting}~\cite{spirakis}\end{tabular}&$1$ &&\checkmark&$O(d_{\max}^n)$&\checkmark&$O(n)$&\\
[.1in]
\hline
\rule{0pt}{4ex}
\emph{Conscious}~\cite{conscious}&$1$ &\checkmark&\checkmark&$n$&\checkmark&\begin{tabular}{c}$O(e^{N^2}N^3) \Rightarrow$\\ $O(e^{d_{\max}^{2n}}d_{\max}^{3n})$ using~\cite{spirakis}\end{tabular}&\\
[.1in]
\hline
\rule{0pt}{4ex}
\emph{Unconscious}~\cite{conscious}&$1$ &&&$n$&No&\begin{tabular}{c}No theoretical\\ bounds\end{tabular}&\\
[.1in]
\hline
\rule{0pt}{4ex}
$\mathcal{A}_{\mathcal{O}^P}$~\cite{oracle}&$1$ &&\begin{tabular}{c}
Oracle\\ for each\\ node\end{tabular}&$n$&Eventually&Unknown&\\
[.1in]
\hline
\rule{0pt}{4ex}
\textsf{EXT}~\cite{LunaB15}&$1$ &&&$n$&\checkmark&$O(n^{n+4})$&EXPSPACE\\
[.1in]
\hline
\rule{0pt}{4ex}
\begin{tabular}{c}\textsc{Incremental}\\\textsc{Counting}~\cite{opodisCounting}\end{tabular}&$1$ &&\checkmark&$n$&\checkmark&$O\left(n\left(2d_{\max}\right)^{n+1}\frac{\ln n}{\ln d_{\max}}\right)$&\\
[.1in]
\hline
\rule{0pt}{4ex}
\begin{tabular}{c}\textsc{Methodical}\\ \textsc{Counting}\\\cite{KowalskiMicalp18}\end{tabular}&$1$ &&&$n$&\checkmark&$O(n^5\ln^2 n)$&PSPACE\\
[.1in]
\hline
\hline
\rule{0pt}{4ex}
\begin{tabular}{c}\textsc{Methodical}\\ \textsc{multi-Counting}\\$[$This work$]$\end{tabular}&$\ell\geq 1$ &&&$n$&\checkmark&\begin{tabular}{c}$O((n^{4+\epsilon}/\ell)\log^3 n)$\\ for any $\epsilon>0$\end{tabular}&PSPACE\\
[.1in]
\hline
\rule{0pt}{4ex}
\begin{tabular}{c}\nameR\\$[$This work$]$\end{tabular}&$0$ &&&\begin{tabular}{c}$n$ \\ prob. $\ge 1-\zeta$\end{tabular}&No&\begin{tabular}{c}$O((n+1/\zeta)^{4+\epsilon}\log^3 (n+1/\zeta))$\\ for any $\epsilon>0$ and $\zeta>0$\end{tabular}&PSPACE\\
[.1in]
\hline
\end{tabular}
}
\end{table}

\section{Previous Work}
\label{relwork}
A comprehensive overview of work related to ADNs can be found in a survey by Casteigts et al.~\cite{arnaudSurvey} and references in the papers cited here.

With respect to lower bounds, it was proved in~\cite{baldoni} that at least $\Omega(\log n)$ rounds are needed, even if $D$ is constant. Also, 
$\Omega(\mathcal{D})$ is a lower bound since at least one node needs to hear about all other nodes to obtain the right count.

Counting and  \emph{Naming} was already studied in~\cite{spirakis} 
for dynamic and static networks, showing 
that it is impossible to solve Counting without the presence of a distinguished node, even if nodes do not move. 
The Counting protocol 
requires knowledge of an upper bound on the dynamic maximum degree of the network $d_{\max}$, and 
obtains
only an upper bound, 
which may be as bad as exponential. 

Conscious Counting~\cite{conscious} computes the exact count, but it needs to start from an upper bound on the size, and it takes exponential time only if such upper bound is tight up to constants. 
In the same work and follow-up papers~\cite{oracle,experimentalConscious}, 
more challenging scenarios where $d_{\max}$ is unknown are studied,
but protocols either do not terminate~\cite{conscious}, 
or the protocol is terminated heuristically~\cite{experimentalConscious}. In experiments~\cite{experimentalConscious}, such heuristic was found to perform well on dense topologies, but for other topologies the error rate was high. 
Another protocol in~\cite{oracle} is shown to terminate eventually, without running-time guarantees and under the assumption of having for each node an estimate of the number of neighbors in each round. In~\cite{spirakis} it was conjectured that some knowledge of the network  such as the latter would be necessary, but the conjecture was disproved later in~\cite{LunaB15}. On the other hand the protocol in~\cite{LunaB15} requires exponential space. 

Incremental Counting, presented recently in~\cite{opodisCounting},
reduced exponentially the best-known running time guarantees.
The protocol obtains the exact count, all nodes terminate simultaneously, the topology dynamics is only limited to $1$-interval connectivity, it only requires polynomial space, and it only requires knowledge of 
$d_{\max}$. 
The running time is still exponential, but reducing from doubly-exponential
was an important step towards understanding the complexity of Counting. 

In a follow-up paper~\cite{netysCounting}, Incremental Counting was tested experimentally showing a promising polynomial behavior. The study was conducted on pessimistic inputs designed to slow the convergence, such as bounded-degree trees rooted at the leader uniformly chosen at random for each round, and a single path starting at the leader with non-leader nodes permuted uniformly at random for each round. The protocol was also tested on static versions of the inputs mentioned, classic random graphs, and networks where some disconnection is allowed. 
The results showed that,
even for topologies that stretch the dynamic diameter, the running times obtained are below $d_{\max} n^3$. It was also observed that random graphs, as used in previous experimental studies~\cite{experimentalConscious}, reduce the convergence time, and therefore are not a good choice to indicate worst-case behavior. These experiments showed good behavior even for networks that sometimes are disconnected, indicating that more relaxed models of dynamics
(e.g.~\cite{FCFMMZ:randomgeocast,geocast}) are worth studying.

All in all, the experiments in~\cite{netysCounting} showed that Incremental Counting behaves well in a variety of pessimistic inputs, but not having a proof of what a worst-case input looks like, 
and because the experiments were restricted
to a range of values of $n$ far from the expected massive size of an ADN, a theoretical proof of polynomial time remained an open problem even from a practical perspective.

A polynomial Counting algorithm was presented in a manuscript~\cite{BaldoniTR}, relying on the availability of an algorithm to compute 
the average of input values, one at each node.
Such average computation is modeled as a Markov chain with underlying doubly-stochastic matrix, which requires topology information within two hops (cf.~\cite{nedic2009}). In the pure model of ADN, such information is not available, and gathering it may not be possible due to possible topology changes from round to round.

Recently, we presented the first polynomial-time deterministic Counting algorithm for ADNs in~\cite{KowalskiMicalp18}, called~\MC. 
Unlike previous works, \MC does not require any knowledge of network characteristics, such as dynamic maximum degree or an upper bound on the size. That is, it works in the pure model of ADN. 
Like previous works, \MC requires the presence of a distinguished node.
In the present work, we generalize that assumption assuming the presence of $\ell\geq 1$ distinguished nodes.   
As in~\cite{KowalskiMicalp18}, we leverage previous work on lazy random walks to analyze \nameD, but the alarms to detect wrong computations have been completely re-designed to deal with multiple distinguished nodes. Moreover, with respect to \MC, \nameD achieves a $\Omega((\log^2 n)/(n\ell))$ speed-up (cf. Table~\ref{table}). That is, even for $\ell=1$, \nameD also provides a speed-up with respect to previous work. 

In the same paper~\cite{KowalskiMicalp18}, we also presented extensions of \MC to compute more complex functions, such as the sum of input values held by nodes and other algebraic and Boolean functions.

\miguel{
With respect to randomized Counting, a linear Counting algorithm  for dynamic networks was presented in~\cite{KuhnLO2010}. The algorithm requires unique identifiers (i.e., it is not applicable to ADNs), knowledge of an upper bound on the size of the network, and only guarantees an approximation to the network size. To the best of our knowledge, no randomized Counting algorithms for ADNs have been studied before. 
}

Other studies also dealing with the time complexity of information gathering exist~\cite{
chen2013role,
banerjee2014epidemic,
sanghavi2007gossiping,
boyd2006randomized,
mosk2008fast,
sarma2015distributed}, 
but include in their model additional assumptions, such as the network having the same topology frequently enough or node identifiers. 

\section{Impossibility of Counting}
\label{sec:impossibility}

In this section, we show impossibility results for Counting in \ADNs. 
After showing the existence of adversarial networks in a technical lemma, 
we establish that deterministic Counting is not possible in \ADNs without knowing the number of distinguished nodes. 
For randomized algorithms, we show that
for any $\ell\geq 0$, there exists an ADN such that, if $\ell = 0$ or $\ell$ is unknown to the nodes, there is no randomized algorithm that with constant probability solves Counting.
These results complement the previously known impossibility for deterministic algorithms without distinguished nodes~\cite{spirakis}. 
Note that our impossibility results hold even for static anonymous networks.

\subsection{Deterministic Algorithms}


\begin{lemma}
\label{lem:impossibility}
For every positive integer $\lambda$ there are two networks
of $n=\lambda+4$ and $n=2(\lambda+4)$ nodes, respectively, such that: 
	no deterministic algorithm could successfully accomplish Counting
	on both of them in finite time, even if some $\ell=\lambda$
		nodes are \ldr in the former and $\ell=2\lambda$ nodes
		are \ldr in the latter network.
\end{lemma}

\begin{proof}
We prove the claim showing two different networks of different size such that 
any algorithm 
executed in these two networks cannot distinguish between 
these two  
executions. Hence, it cannot obtain the correct count in both.

Throughout the proof, solely for the purpose of definition of networks and analysis
of executions of the algorithms considered,
we distinguish between nodes; this obviously is not possible
for the algorithms due to the definition of the ADN communication model
and the deterministic nature of the algorithms, in the sense that when receiving messages and changing state a node does not know from which specific nodes they have arrived.

Consider a positive integer $\lambda$ and the two associated networks
$G_{\lambda,1}$ and $G_{\lambda,2}$ defined as follows, c.f., Figure~\ref{fig:}.

\begin{figure}[ht]
  \centering
\begin{subfigure}{0.45\textwidth}
  \centering
\psfrag{l}{$\lambda$ \ldr}
\psfrag{n}{nodes}
\includegraphics[width=\textwidth]{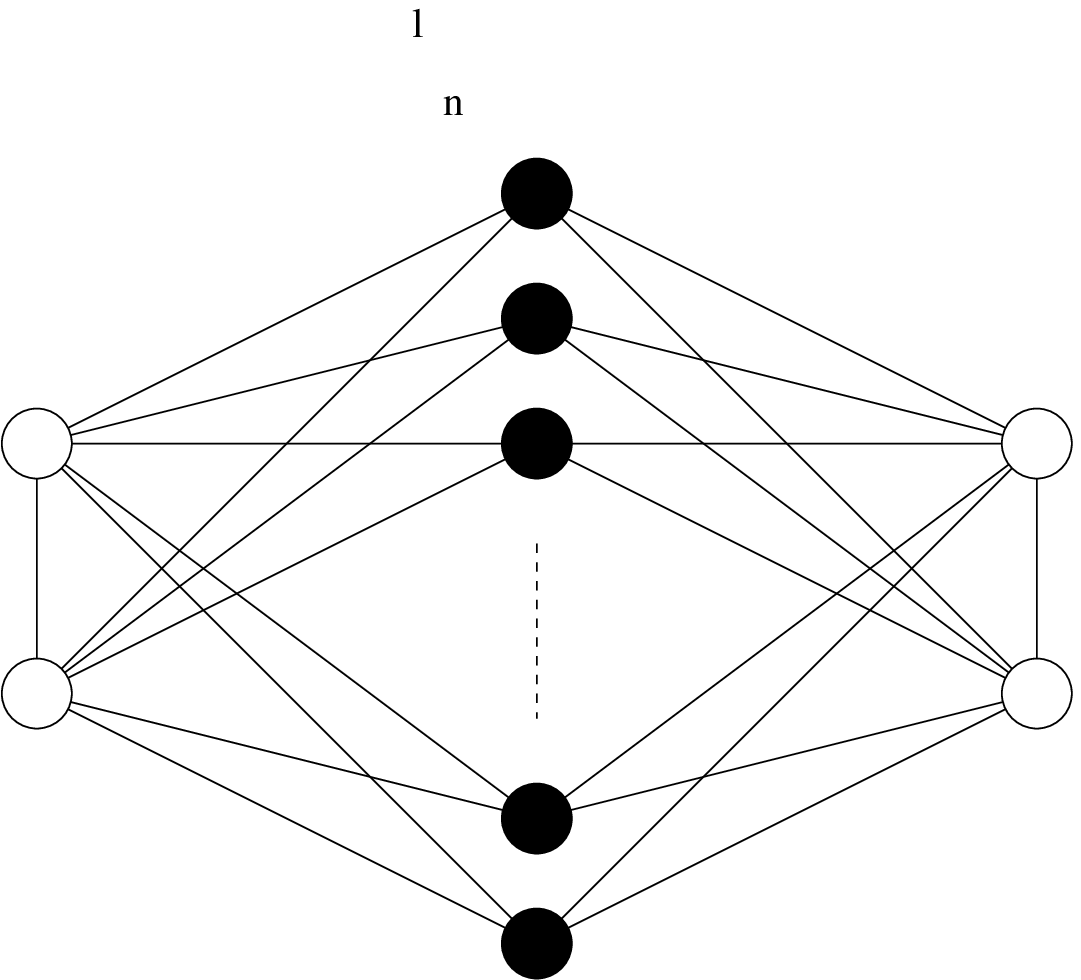}
  \caption{$G_{\lambda,1}$}
  \label{fig:Gell1}
\end{subfigure}
\hspace{.3in}
\begin{subfigure}{.45\textwidth}
  \centering
\psfrag{l}{$\lambda$ \ldr}
\psfrag{n}{nodes}
\psfrag{4}{$4$ connected}
\psfrag{p}{pairs of}
\psfrag{w}{\notldr nodes}
\includegraphics[width=\textwidth]{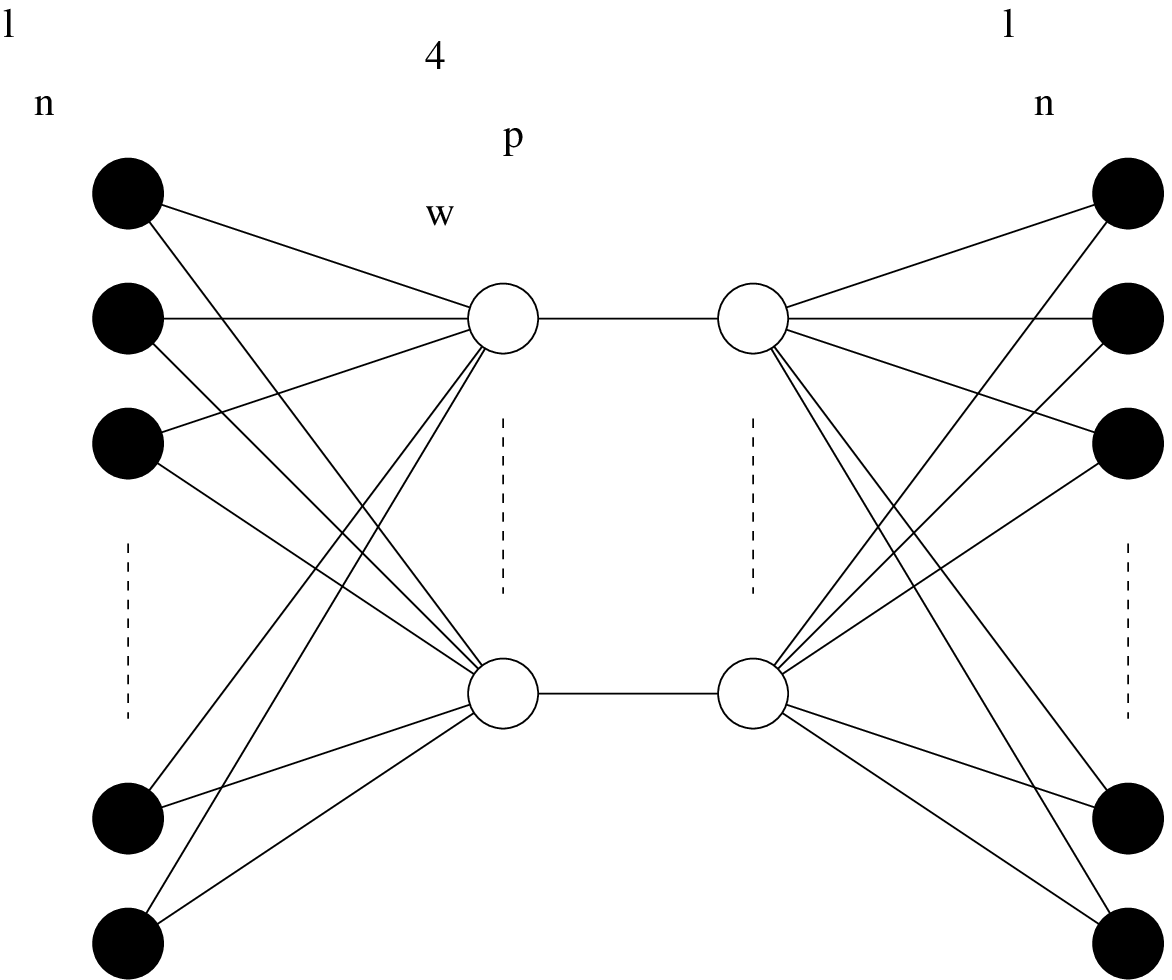}
  \caption{$G_{\lambda,2}$}
  \label{fig:Gell2}
\end{subfigure}
\caption{Illustration of Lemma~\ref{lem:impossibility}.}
\label{fig:}
\end{figure}

$G_{\lambda,1}$ consists of $\lambda$ \ldr nodes and $4$ \notldr nodes,
connected in such a way that some two of the \notldr nodes 
form a cycle of length $3$ with each of the \ldr nodes,
and the remaining two \notldr nodes also 
form a cycle of length $3$ with each of the \ldr nodes.
It follows that the graph is connected. 
Each \ldr node has $4$ neighbors (namely, all \notldr nodes),
and each \notldr node has $\lambda+1$ neighbors (specifically,
$\lambda$ \ldr nodes and one \notldr node).

$G_{\lambda,2}$ consists of $2\lambda$ \ldr nodes and $8$ \notldr nodes.
The \ldr nodes are grouped into two groups of size $\lambda$ each,
while the \notldr nodes are grouped into four pairs.
Each \ldr node in the first group is connected with all first elements
of the pairs of \notldr nodes.
Similarly, each \ldr node in the second group is connected with all second elements
of the pairs of \notldr nodes.
Additionally, for each pair of \notldr nodes, both nodes in the pair are also connected.
It is easy to check that the graph is connected, 
each \ldr node has four neighbors (namely, either all \notldr nodes
at first position in pairs or all \notldr nodes
at second position in pairs, depending on the group it belongs to),
and each \notldr node has $\lambda+1$ neighbors (specifically,
$\lambda$ \ldr nodes, either from the first or the second group
depending on whether the node is first or second in its pair, 
and one \notldr node from the same pair).

In order to prove 
the lemma, 
consider executions of a deterministic Counting algorithm on $G_{\lambda,1}$
and on $G_{\lambda,2}$. 
We prove that the following invariant holds for each communication round $t$:
\begin{itemize}
\item
All \ldr nodes in graphs $G_{\lambda,1}$ and $G_{\lambda,2}$ have the same history (of received messages and input value) in the beginning of round $t$;
\item
All \notldr nodes in graphs $G_{\lambda,1}$ and $G_{\lambda,2}$ have the same history (of received messages and input value) in the beginning of round $t$.
\end{itemize}
The proof is by induction. In the beginning of round one
all \ldr nodes in both graphs know only that they
are \ldr, and similarly \notldr nodes
in both graphs know only that they are not \ldr. Hence
the invariant holds. 

Suppose the invariant holds for some $t\ge 1$, we show
that it is also true for $t+1$. 
Each \ldr node in both graphs is a neighbor
of four \notldr nodes.
By the invariant for $t$, they had the same history at the beginning
of round $t$ and, since they follow a deterministic algorithm,
they all either broadcast exactly the same message or
do not broadcast in round $t$.
Hence, at the end of round $t$ each \ldr node
gets the same messages (or no message, respectively), and since
by the invariant for $t$ they all had the same history at the
beginning of round $t$, their histories are also the same
at the end of round $t$, that is, also at the beginning
of round $t+1$. Hence, the first bullet of the invariant for $t+1$
holds.

Similarly, we argue about \notldr nodes.
Every \notldr node in both graphs is a neighbor
of one \notldr node and $\lambda$ \ldr nodes.
By the invariant for $t$, all these \ldr neighbors had the same history at the beginning
of round $t$, 
and the same applies to the \notldr neighbors.
Since they all follow a deterministic algorithm,
all \ldr neighbors either broadcast exactly the same message or do not broadcast in round $t$;
similar property holds for \notldr neighbors.
Hence, at the end of round $t$ each \notldr node
gets the same set of messages (or no message, if both \ldr and \notldr nodes do not transmit in round $t$), and since
by the invariant for $t$ they all had the same history at the
beginning of round $t$, their histories are also the same
at the end of round $t$, which means, also at the beginning
of round $t+1$. Hence the second bullet of the invariant for $t+1$
holds as well.

To complete the proof, 
note that if nodes in both networks
had accomplished Counting at some round $t$, 
\ldr nodes in $G_{\lambda,1}$ would return $\lambda+4$
while \ldr nodes in $G_{\lambda,2}$ would return $2(\lambda+4)$, all based on the same history, which is impossible. More formally, consider the first time $t$
in which some \ldr node in any of these two networks
returns the correct number of nodes. By the invariant, all \ldr nodes in both networks have the same history, 
and since they follow the same deterministic algorithm they all return the same value as the number of nodes. This however
is a contradiction, since \ldr nodes in the first network should return a different value. Hence, such
time $t$ does not exist and no deterministic algorithm
accomplishes Counting in a finite time.

\end{proof}

The following follows from the above lemma.

\begin{theorem}
If the number of \ldr nodes is not given as a parameter,
then 
there is no deterministic Counting algorithm where all nodes stop and return the correct count in all ADNs, 
not even with an approximation smaller than $\sqrt{2}$. 
\end{theorem}

\begin{proof}
The impossibility of obtaining an exact count in all networks follows from Lemma~\ref{lem:impossibility}.
Regarding the approximation, since the algorithm could not distinguish between the two networks in Lemma~\ref{lem:impossibility} of sizes $\lambda+4$ and $2(\lambda+4)$ respectively, for any $\lambda>0$, even if the algorithm stopped and returned a consistent value (across all nodes) it would be within at least 
$\sqrt{2}$ factor from some of these two sizes $\lambda+4$ and $2(\lambda+4)$;
since any of these two is feasible in the execution, the approximation is at least $\sqrt{2}$,
which occurs if the algorithm outputs a count $x$ such that $x/(\lambda+4) = 2(\lambda+4)/x$.
\end{proof}

Remark: the inapproximability result could be extended from $\sqrt{2}$ to {\em any} constant by considering networks with $c\cdot (\lambda+4)$ nodes, for an arbitrary constant $c$ instead of $c=2$ used in the above proofs for the ease of arguments.


\subsection{Randomized Algorithms}

\begin{theorem}
\label{thm:randimposs}
For any constant $0<c<1$ and any $\ell\geq 0$, 
there exists an \ADN such that, 
if $\ell=0$ or $\ell$ is unknown to the nodes, 
there is no Counting randomized algorithm such that, with probability $c$,
each node outputs the correct count when it stops. 
\end{theorem}

\begin{proof}
For the sake of contradiction, assume such an algorithm exists, call it $\cA$.

For an arbitrary number of \ldr nodes $\ell\geq 0$, consider the following network that we call a \emph{\gadget}, depicted in Figure~\ref{fig:gadget}.
If $\ell=0$, a \gadget is a simple path of $n$ \notldr nodes connected by $n-1$ links, for arbitrary $n>1$.
If on the other hand it is $\ell>0$, a \gadget is also a path of \notldr nodes, where each \notldr node is linked to a separate \ldr node (a caterpillar). That is, the number of \notldr nodes is also $\ell$ and the total number of nodes in the \gadget is $n=2\ell$. 

\begin{figure}[ht]
  \centering
\begin{subfigure}{0.45\textwidth}
  \centering
\psfrag{n}{$n$ nodes}
\includegraphics[width=\textwidth]{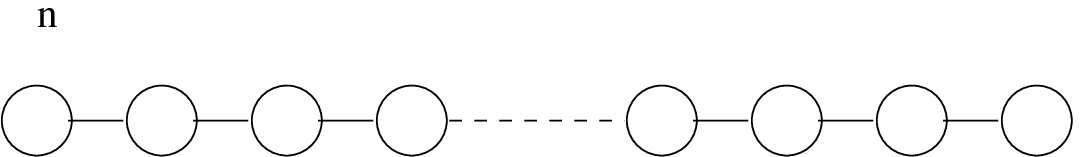}
  \caption{$\ell=0$}
\end{subfigure}
\hspace{.2in}
\begin{subfigure}{.45\textwidth}
  \centering
\psfrag{l}{$\ell$ nodes}
\includegraphics[width=\textwidth]{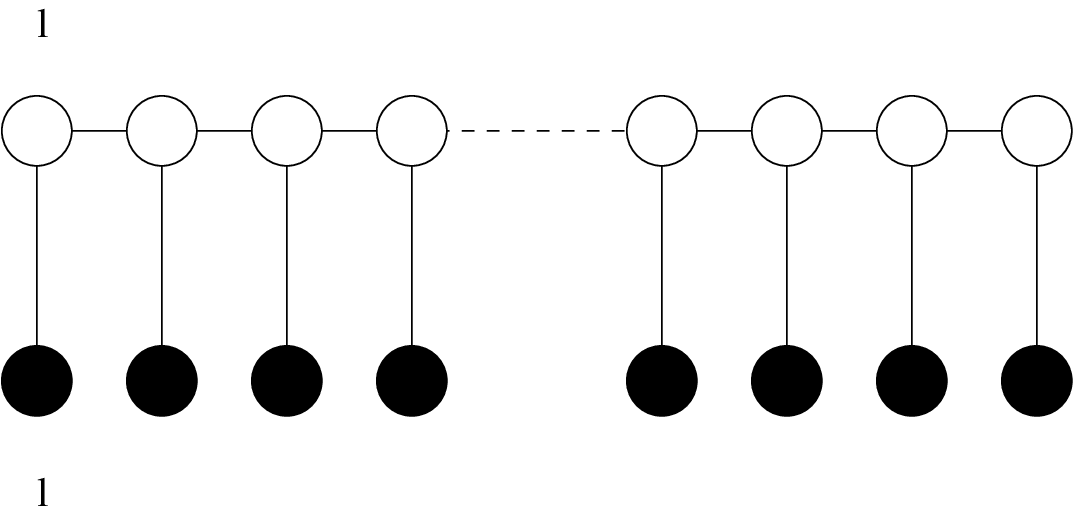}
  \caption{$\ell>0$}
\end{subfigure}
\caption{Illustration of a \gadget.}
\label{fig:gadget}
\end{figure}

Let $G_{1,n,\ell}$ be a network formed by one \gadget with an additional edge connecting the (\notldr) end nodes of the gadget as a cycle (see Figure~\ref{fig:1cycle} for $\ell>0$).
It follows that every node has the same number of neighbors of each color.  
With probability at least~$c$, all nodes running Algorithm $\cA$ on $G_{1,n,\ell}$ must stop and output the correct count by some time $T(n,\ell)$.

\begin{figure}[ht]
  \centering
\psfrag{l}{$\ell$ nodes}
\includegraphics[width=0.45\textwidth]{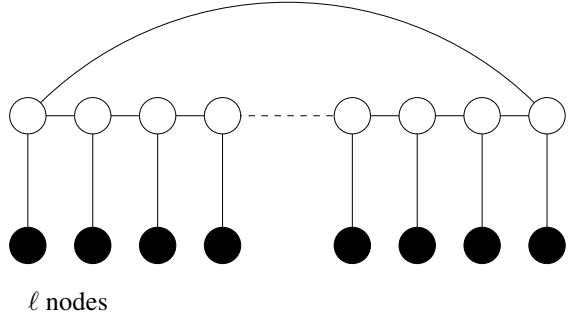}
  \caption{$G_{1,n,\ell}$ for $\ell>0$.}
  \label{fig:1cycle}
\end{figure}

Without loss of generality, assume that in $\cA$ each node draws one random bit per round of communication. 
(If more random bits per round are used, the same argument can be extended to more outcomes.) 

Consider an execution of $\cA$ on $G_{1,n,\ell}$ in the first $T(n,\ell)$ rounds.
Starting from the initial state where nodes have no information besides their own color (recall that nodes do not have identifiers), in each round a node makes decisions, based on the random bits drawn and the received states of its neighbors to move to another state. With respect to the random bits drawn, by time $t$ a node may be in one of $2^t$ states. 
We call the states of all network nodes at a given time a \emph{configuration} of states. 
A configuration in $G_{1,n,\ell}$ where $\cA$ stops at all nodes with an exact count of $n$ is called a \emph{winning configuration}. 
By definition of $\cA$, the probability of ending at a winning configuration is at least $c$, and there are $2^{T(n,\ell)n}$ possible configurations. Hence, there must exist some winning configuration $\Gamma$ that occurs in $\cA$ with probability at least $c/2^{nT(n)}$.
Denote by $\Gamma_{\to t}$ the part of configuration $\Gamma$ by round $t$.

Consider an execution of $\cA$ on a network $G_{x,n,\ell}$, where $x$ will be defined later, formed by $x$ \gadgets connected in sequence by the (\notldr) end points so that \notldr nodes form a cycle (see Figure~\ref{fig:xcycle}).

\begin{figure}[ht]
  \centering
\psfrag{g}{gadget}
\psfrag{x}{$x$ gadgets}
\includegraphics[width=\textwidth]{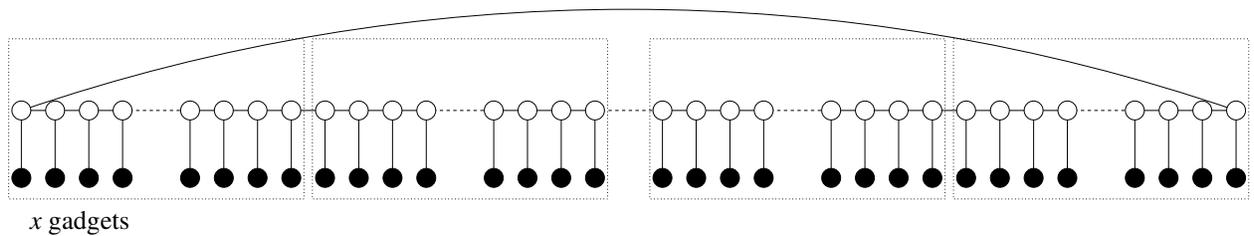}
  \caption{$G_{x,n,\ell}$ for $\ell>0$.}
  \label{fig:xcycle}
\end{figure}

Let a path of $2\lceil T(n,\ell)/(n-\ell)\rceil + 1$ \gadgets in $G_{x,n,\ell}$ be called a \emph{witness}, 
where the \gadget in the center of a witness is called the \emph{core},
and the $\lceil T(n,\ell)/(n-\ell)\rceil$ \gadgets in each side of the core are called a \emph{buffer}
(see Figure~\ref{fig:witness}).

\begin{figure}[ht]
  \centering
\psfrag{T}{$\lceil T(n,\ell)/(n-\ell)\rceil$ \gadgets}
\psfrag{1}{1 \gadget}
\psfrag{c}{core}
\psfrag{b}{buffer}
\includegraphics[width=\textwidth]{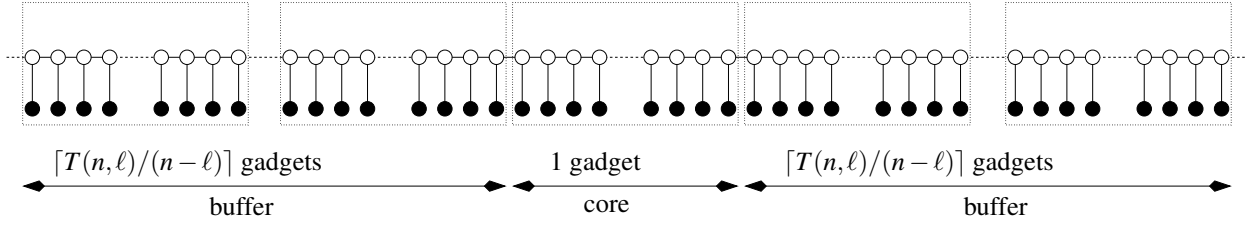}
  \caption{A witness in $G_{x,n,\ell}$ for $\ell>0$.}
  \label{fig:witness}
\end{figure}

We show now that, for $x$ large enough, after executing $\cA$ on $G_{x,n,\ell}$ for $T(n,\ell)$ rounds, the core of some witness has the configuration $\Gamma$, hence nodes in that core stop with an incorrect count of $n$ with probability larger than $1-c$, which proves the theorem.

Let $x$ be such that witnesses are separated by $2\lceil T(n,\ell)/(n-\ell)\rceil$ \gadgets. That is, they are disjoint and between any pair of consecutive witnesses there are $2\lceil T(n,\ell)/(n-\ell)\rceil$ \gadgets that do not belong to any considered witness.
Therefore, during the first $T(n,\ell)$ rounds of the execution of $\cA$ on $G_{x,n,\ell}$, the configurations on witnesses are independent.

Consider a single witness in $G_{x,n,\ell}$. 
Let any connected pair of \notldr and \ldr nodes be called a \emph{pair}.
Define a $t$-semi-core of the witness, for $0\leq t\leq \lceil T(n,\ell)/(n-\ell)\rceil (n-\ell)$, as a set containing the pairs in the core and all the pairs in the buffer, such that their \notldr nodes are located at distance at most $T(n,\ell)-t$ from the core. In particular, the $0$-semi-core is the whole witness, and the $\lceil T(n,\ell)/(n-\ell)\rceil (n-\ell)$-semi-core is the core (see Figure~\ref{fig:semicore}).

\begin{figure}[ht]
  \centering
\psfrag{c}{core}
\psfrag{b}{buffer}
\psfrag{Tcore}{$\mathcal{T}$-semi-core}
\psfrag{tcore}{$t$-semi-core}
\psfrag{t}{$t$ pairs}
\psfrag{0core}{$0$-semi-core}
\includegraphics[width=\textwidth]{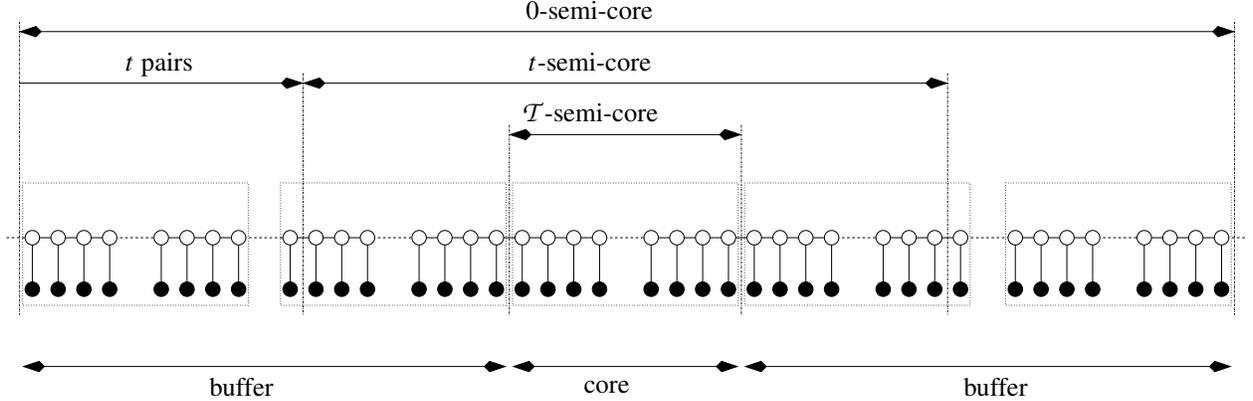}
  \caption{Semi-cores in a witness in $G_{x,n,\ell}$, for $\ell>0$. $t$ is such that $0\leq t\leq \mathcal{T}$ and $\mathcal{T}=\lceil T(n,\ell)/(n-\ell)\rceil (n-\ell)$.}
\label{fig:semicore}
\end{figure}

Using the above notation, we prove the following invariant (illustrated in Figure~\ref{fig:invariant}).

\begin{itemize}
\item[]
\emph{For any $0\leq t\leq \lceil T(n,\ell)/(n-\ell)\rceil (n-\ell)$, 
with probability at least $c/2^{nt}$, any pair $(v,v')$ in the $t$-semi-core located at distance $x\leq \lceil T(n,\ell)/(n-\ell)\rceil (n-\ell)-t$ from the core, has the same configuration in $\Gamma_{\to t}$ as the pair $(w,w')$ in the core at distance $x\bmod (n-\ell)$ from the other buffer.} 
\end{itemize}

\begin{figure}[ht]
  \centering
\psfrag{c}{core}
\psfrag{b}{buffer}
\psfrag{t}{$t$}
\psfrag{v}{$v$}
\psfrag{v'}{$v'$}
\psfrag{w}{$w$}
\psfrag{w'}{$w'$}
\psfrag{x}{$x$}
\psfrag{xmod}{$x \mod (n-\ell)$}
\includegraphics[width=\textwidth]{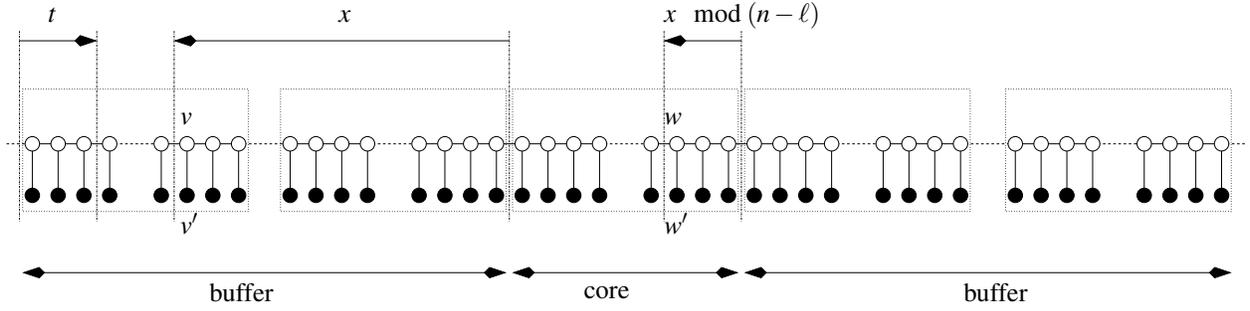}
  \caption{Invariant in a witness in $G_{x,n,\ell}$, for $\ell>0$. $0\leq t\leq \mathcal{T}$, $0\leq x\leq \mathcal{T}-t$, and $\mathcal{T}=\lceil T(n,\ell)/(n-\ell)\rceil (n-\ell)$.}
\label{fig:invariant}
\end{figure}

We prove the invariant by induction as follows.
When the execution of $\cA$ starts, all pairs in the $0$-semi-core are in the same state. 
Then, assuming inductively that the invariant holds for some round $0\leq t< \lceil T(n,\ell)/(n-\ell)\rceil (n-\ell)$,
we can extend the configuration of the $t$-semi-core at all but the end pairs of this semi-core
to satisfy the invariant for $t+1$. This is because all these pairs receive messages from other nodes satisfying the invariant
for $t$. Thus, cyclically modulo $(n-\ell)$ from a core boundary, nodes mimick 
the behavior of $\cA$ on $G_{1,n,\ell}$ at round $t+1$ of the protocol, which leads to the partial configuration $\Gamma_{\to t+1}$ on $G_{1,n,\ell}$, and thus on the corresponding nodes of the $(t+1)$-semi-core, cyclically modulo $(n-\ell)$ from a core boundary.

It follows from the invariant for $T(n,\ell)$ that the $n$ nodes in the core end up in configuration $\Gamma$.
Since $\Gamma$ is a winning configuration on $n$ nodes, all nodes in the core stop and output a count of $n$, 
which violates the correctness during the considered execution of $\cA$ on $G_{x,n,\ell}$.

It remains to prove that the union of the above events, over all the considered $x/(4\lceil T(n,\ell)/(n-\ell)\rceil+1)$ witnesses, holds with probability larger than $1-c$. 
For $y=x/(4\lceil T(n,\ell)/(n-\ell)\rceil+1)$, we show that the complementary event holds with probability smaller than $c$, that is,
$\left(1-\left(c/2^{nT(n,\ell)}\right)^2\right)^y < c$. 
This is indeed implied by
$\exp\left(-c^2y/2^{2nT(n,\ell)}\right) < c$, which holds for
$y > \ln (1/c) 2^{2nT(n,\ell)} / c^2$.
That is, for $x=\left(1+\ln (1/c) 2^{2nT(n,\ell)} / c^2\right) (4\lceil T(n,\ell)/(n-\ell)\rceil+1)$, with probability larger than $1-c$, there exists at least one witness in $G_{x,n,\ell}$ with a winning configuration $\Gamma$ after running $\cA$ for $T(n,\ell)$ rounds. 
Hence, the core nodes stop and output an incorrect count, and the claim follows.
%
%
\end{proof}

Since there is no deterministic algorithm that solves Counting if the number of \ldr nodes is unknown, and Counting cannot be solved even with a constant probability by randomized algorithms if the number of \ldr nodes is unknown or there is no \ldr node, we pursue two directions.
One is to design a polynomial-time deterministic algorithm for exact Counting if the exact number of \ldr nodes is a priori known.
The other, for randomized algorithms, is to focus on Unconscious Counting~\cite{conscious} as defined in Section~\ref{prelim}. 


\section{Deterministic Counting}
\label{algorithm}

%
%
%
%
%

In this section we present and analyze our deterministic Counting algorithm \nameD. 

The underlying idea of the algorithm is to search for the number of nodes testing different estimates until the correct value is found. Starting with an initial value, the estimate is doubled after each failed attempt. If at some point the estimate goes above $n$, the correct estimate is found by binary search in the 
range between the previous and the current estimate. To organize this exploration, we structure the execution of the algorithm in \emph{epochs}. In each epoch, a new estimate is tested. Once the correct estimate is found, the execution stops.

In each epoch, to detect whether the estimate is low, correct, or high, we use some alarms, most of them based on the value of some \emph{potentials} held at nodes. 
To implement these alarms we use a gossip-based approach to collect at the \ldr nodes the potentials, which are initially held only by the \notldr nodes.
The way that potentials change, and their convergence values, depend on the correctness of the estimate.  

For the above potential-collection process we could have the \ldr nodes collecting continuously until done, but to facilitate our analysis techniques we need all nodes participating in the gossiping process. Thus, rather than collecting all the potential at once, we further structure each epoch in \emph{phases} such that, during each phase, all nodes average the current potentials, and only at the end of the phase \ldr nodes remove their current potential to a separate accumulator.
The number of rounds of communication in each phase are enough to reach such average or detect a wrong estimate. 

In the following paragraphs, we give further details referring to the pseudocode in Algorithms~\ref{leaderAlg} and~\ref{otherAlg}.

\begin{algorithm}[htbp]
\caption{\nameD algorithm for each {\bf\emph{\ldr node}}. $N$ is the set of neighbors of this node in the current round. $\ell$ is the number of \ldr nodes. The parameters $d,p,r$ and $\tau$ are as defined in Theorem~\ref{thm:many}.}
\label{leaderAlg}
\DontPrintSemicolon
	$k \gets \ell+1, min\gets k, max\gets\infty$ \tcp*{initial size estimate and range}
	\Repeat(\tcp*[f]{iterating epochs}){$status=done$}{  \label{epochsleader}
		$status\gets probing$ \tcp*{status$=$probing$|$low$|$high$|$done}
		$\Phi\gets 0$ \tcp*{current potential}
		$\rho\gets 0$ \tcp*{potential accumulator}
		\For(\tcp*[f]{iterating phases}){$phase=1$ to $p$}{  \label{phasesleader}
			\For( \tcp*[f]{iterating rounds}){$round=1$ to $r$}{ \label{roundsleader}
				Broadcast $\langle\Phi,status\rangle$ 
				and Receive $\langle\Phi_i,status_i\rangle, \forall i\in N$ \;
				\If{$status=probing$ {\bf and} $|N|\leq d-1$ {\bf and} $\forall i\in N:status_i=probing$} {
					$\Phi\gets \Phi + \sum_{i\in N}\Phi_i/d - |N|\Phi/d$ \label{potupdate} \tcp*{update potential}	
				}
				\Else(\tcp*[f]{$k<n$}){  \label{leadertoomany}
					$status\gets low$, \label{alarminsecondleader}
					$\Phi\gets \ell$\;
				}
			}
			\lIf(\tcp*[f]{$k<n$}){$phase=1$ {\bf and} $\Phi> \tau$}{ 
					$status\gets low$,
					$\Phi\gets \ell$  \label{leaderthreshold}
			}
			\If(\tcp*[f]{prepare for next phase}){$status=probing$}{ 
				$\rho \gets \rho + \Phi$ \label{consume} \label{rhoupdate} \tcp*{consume potential} 
				$\Phi \gets 0$ \; \label{phireset}
			}
		}
		\If{$status=probing$}{ \label{statuscheck}
			\If(\tcp*[f]{$k=n$}){$(k-\ell)(1-k^{-\gamma})\leq \rho \leq (k-\ell)(1+k^{-\gamma})$}{ 
				$status\gets done$} \label{range}
			\lIf(\tcp*[f]{$k>n$}){$\rho < (k-\ell)(1-k^{-\gamma})$}{ 
				$status\gets high$}  \label{toobig}
			\lIf(\tcp*[f]{$k<n$}){$\rho > (k-\ell)(1+k^{-\gamma})$}{ 
				$status\gets low$}  \label{toolow}
		}
		\For(\tcp*[f]{disseminate status}){$round=1$ to $d$}{  
			Broadcast $\langle status\rangle$
			and Receive $\langle status_i\rangle, \forall i\in N$ \label{leadernotification}\;
		} 
		\eIf(\tcp*[f]{prepare for next epoch}){$status=low$}{  \label{leaderupdate}
			$min\gets k+1$ \;
			\leIf{$max=\infty$}{$k\gets 2k$}{$k\gets\lfloor(min+max)/2\rfloor$}
		}{
			\If{$status=high$}{
		         	$max\gets k-1$ \;
				$k\gets\lfloor(min+max)/2\rfloor$ \;
			}
		}
	}
	\textbf{return} $k$ \;

\end{algorithm}

\begin{algorithm}[htbp]
\caption{\nameD algorithm for each {\bf\emph{\notldr node}}. $N$ is the set of neighbors of this node in the current round. $\ell$ is the number of \ldr nodes. The parameters $d,p,r$ and $\tau$ are as defined in Theorem~\ref{thm:many}.}
\label{otherAlg}
\DontPrintSemicolon
	$k \gets \ell+1, min\gets k, max\gets\infty$ \tcp*{initial size estimate and range}
	\Repeat(\tcp*[f]{iterating epochs}){$status=done$}{  \label{epochsother}
		$status\gets probing$ \tcp*{status$=$probing$|$low$|$high$|$done}
		$\Phi\gets \ell$ \tcp*{current potential}
		\For(\tcp*[f]{iterating phases}){$phase=1$ to $p$}{ \label{phasesother}
			\For(\tcp*[f]{iterating rounds}){$round=1$ to $r$}{  \label{roundsother}
				Broadcast $\langle\Phi,status\rangle$
				and Receive $\langle\Phi_i,status_i\rangle, \forall i\in N$ \;
				\If{$status=probing$ {\bf and} $|N|\leq d-1$ {\bf and} $\forall i\in N:status_i=probing$}{ 
					$\Phi\gets \Phi + \sum_{i\in N}\Phi_i/d - |N|\Phi/d$ \label{newpot}
					\tcp*{update potential}
				}	
				\Else(\tcp*[f]{$k<n$}){  \label{othertoomany}
					$status\gets low$, \label{alarminsecondother}
					$\Phi\gets \ell$ \;
				}
			} 
			\lIf(\tcp*[f]{$k<n$}){$phase=1$ {\bf and} $\Phi> \tau$}{  
					$status\gets low$,
					$\Phi\gets \ell$ \label{otherthreshold}
			}
		}
		\For(\tcp*[f]{disseminate status}){$round=1$ to $d$}{  
			Broadcast $\langle status\rangle$
			and Receive $\langle status_i\rangle, \forall i\in N$ \label{othernotification}\;
			\lIf{$\exists i\in N:status_i\neq probing$}{$status\gets status_i$}
		}
		\If(\tcp*[f]{prepare for next epoch} ){$status=low$}{ \label{otherupdate}
			$min\gets k+1$ \;
			\leIf{$max=\infty$}{$k\gets 2k$}{$k\gets\lfloor(min+max)/2\rfloor$}
		}	
		\Else{ 
			\If{$status=high$}{
		         	$max\gets k-1$\;
				$k\gets\lfloor(min+max)/2\rfloor$\;
			}
		}
	}
	\textbf{return} $k$ \;

\end{algorithm}


Initially, each of the $\ell$ \ldr nodes is assigned a potential of $0$ and each of the $n-\ell$ \notldr nodes is assigned a potential of~$\ell$. 
Epoch $k$ corresponds to a size estimate $k$ that is iteratively updated from epoch to epoch until the correct value $n$ is found.
Each epoch is divided into $p$ phases, 
and each phase is composed by $r$ rounds of communication. 

In each round, each node
broadcasts its potential and receives the potential of all its neighbors. Each node keeps only a fraction $1/d$ of the potentials received. The parameters $p$, $r$, and $d$ are functions of $k$. 
The specific functions needed to guarantee correctness and saught efficiency are defined in Theorem~\ref{thm:many}. For the sake of presentation we simply write $p$, $r$, and $d$ instead of $p(k)$, $r(k)$, and $d(k)$.
This varying way of distributing potential is different from previous approaches using mass distribution. 
After communication, each node updates its own potential accordingly (cf. Lines~\ref{leaderAlg}.\ref{potupdate} and~\ref{otherAlg}.\ref{newpot}). That is, it adds a fraction $1/d$ of the potentials received, and subtracts a fraction $1/d$ of the potential broadcasted times the number of potentials received. Then, a new round starts. 
At the end of each phase, each \ldr node ``consumes'' its potential. 
That is, it increases an internal accumulator $\rho$ by its current potential $\Phi$, and $\Phi$ is zeroed for starting the next phase (cf. Lines~\ref{leaderAlg}.\ref{rhoupdate} and~\ref{leaderAlg}.\ref{phireset}). 

The correctness (or incorrectness) of the estimate is detected by various alarms as follows. 
A node stops the update of potential described, raises its potential to $\ell$, and broadcasts an alarm status ``low'' in each round until the end of the epoch if any of the following happens: 
1) at the end of the first phase its potential is above some threshold $\tau$ as defined in Theorem~\ref{thm:many} (cf. Lines~\ref{leaderAlg}.\ref{leaderthreshold} and~\ref{otherAlg}.\ref{otherthreshold}), 
2) at any round it receives more than $d-1$ messages
(cf. Lines~\ref{leaderAlg}.\ref{leadertoomany} and~\ref{otherAlg}.\ref{othertoomany}), or 
3) at any round receives an alarm status ``low'' from one of its neighbors (cf. Lines~\ref{leaderAlg}.\ref{leadertoomany} and~\ref{otherAlg}.\ref{othertoomany}). 
Case 1) allows the \ldr nodes to detect that the estimate is wrong when $k^{1+\epsilon}<n$ for some $\epsilon>0$ (Lemmas~\ref{manyunalarmed} and~\ref{manyalarmsoon}), case 2) allows the \ldr nodes to detect that $d$ is too small and hence the estimate is low, and case 3) allows dissemination of these alarms.
(In ``low'' status the potential is set to $\ell$ to facilitate the analysis, but it is not strictly needed by the algorithm.)

At the end of each epoch, each \ldr node checks the value of $\rho$ and updates its status accordingly. If it is within some range, call it $\Gamma$, the current estimate is correct and each \ldr node changes its status to ``done'' (cf. Line~\ref{leaderAlg}.\ref{range}). Otherwise, the estimate is incorrect. If $\rho$ is below (resp. above) $\Gamma$ each \ldr node changes its status to ``high'' (resp. ``low'') indicating that the estimate is too big (resp. ``low'') (cf. Lines~\ref{leaderAlg}.\ref{toobig} and ~\ref{leaderAlg}.\ref{toolow} respectively). 
The case when $\rho$ is below (resp. above) $\Gamma$ allows to detect when $k>n$ (resp. $k<n\leq k^{1+\epsilon}$) (c.f. Lemmas~\ref{manyksquare} and~\ref{kaboven} respectively.).

After \ldr nodes update their status, the network is flooded with it for $d>k$ rounds (cf. Lines~\ref{leaderAlg}.\ref{leadernotification} and~\ref{otherAlg}.\ref{othernotification}). If $k\geq n$, those rounds are enough for all \notldr nodes to receive the ``done'' or ``high'' status. 
If they receive ``done'', after completing the $k$ rounds all nodes stop. Otherwise, after completing the $k$ rounds all nodes update $k$ according to status to start a new epoch (cf. Lines~\ref{leaderAlg}.\ref{leaderupdate} and~\ref{otherAlg}.\ref{otherupdate}). If $k$ has not been detected to be greater than $n$ since the computation started, it is doubled for the next epoch, otherwise it is updated as in binary search. 
Notice that at the beginning of each epoch the nodes' status are set to ``probing''. Hence, if $k<n$, \notldr nodes may or may not detect that the estimate is low, but if they do not detect it they will move to the next epoch increasing~$k$.

\subsection{Analysis of \nameD}
\label{sec:manyanalysis}

In this section we analyze \nameD. 
References to algorithm lines are given as $\langle algorithm\#\rangle.\langle line\#\rangle$.
We use standard notations $\vec{I}$ for the unit vector, and $L_p$ for the norm of vector $\vec{x}=(x_1,x_2,\dots,x_n)$ as $||\vec{x}||_p = \left(\sum_{i=1}^n |x_i|^p\right)^{1/p}$, for any $p\geq 1$.
Only for the analysis, nodes are labeled as $0,1,2,\dots,n-1$.
The potential of a node $i$ at the beginning of round $s$ of phase $t$ is denoted as $\Phi_{s,t}[i]$, 
the potential of all nodes is denoted as a vector $\vec{\Phi}_{s,t}$, 
and the aggregated potential is then $||\vec{\Phi}_{s,t}||_1$. 
The subindices $s$, $t$, or both are omitted sometimes for clarity.
We will refer to the potential right after the last round of a phase as $\vec{\Phi}_{r+1}$. Such round does not exist in the algorithm, but we use this notation to distinguish between the potential right before \ldr nodes consume their own potential (cf. Line~\ref{leaderAlg}.\ref{rhoupdate}) and the potential at the beginning of the first round of the next phase.

First, we provide a broad description of our analysis of \nameD. 
Consider the vector of potentials $\vec{\Phi}_i$ held by nodes at the beginning of any given phase $i$.
Let a $d$-lazy random walk be a walk that picks each adjacent edge with probability $1/d$, and with the remaining probability it stays at the current vertex.
The way that potentials are updated in each round (cf. Lines~\ref{leaderAlg}.\ref{potupdate} and~\ref{otherAlg}.\ref{newpot}) is equivalent to the progression of a $d$-lazy random walk on the evolving graph underlying the network topology~\cite{michal}, where the initial vector of potentials is equivalent to an initial distribution $\vec{\Pi}_i$ on the overall potential $||\vec{\Phi}_i||_1$
and the probability of choosing a specific neighbor is $1/d$. 

Note that \nameD\ is not a simple ``derandomization'' of the lazy
random walk on evolving graphs. 
First, in the \ADN\ model neighbors cannot be 
distinguished, and even their number is unknown at transmission time
(only at receiving time the node learns the number of its neighbors). 
Second, due to unknown network parameters,
it may happen in an execution of \nameD\ that the total potential received
could be bigger than $1$. 
Third, our algorithm does not know a priori when to terminate and provide a result
even with some reasonable accuracy, as the formulas on mixing and cover time
of lazy random walks depend on the (a priori unknown) number of nodes $n$. 
Nevertheless, we can still use some results obtained in the context of 
analogous lazy random walks
in order to prove useful properties of parts of \nameD,
namely, some parts in which parameters are temporarily fixed and
the number of received messages does not exceed parameter $d$.

It was shown in~\cite{michal} that random walks on $d$-regular explorable evolving graphs have a uniform stationary distribution, and bounds on the mixing and cover time were proved as well. Moreover, it was observed that those properties hold even if the graph is not regular and $d$ is only an upper bound on the degree.\footnote{Their analysis relies on Lemma 12, which bounds the eigenvalues of the transition matrix as long as it is stochastic, connected, symmetric, and non-zero entries lower bounded by $1/d$. Those conditions hold for all the transition matrices, even if the evolving graph is not regular.}

Thus, for the cases where $d$ is an upper bound on the number of neighboring nodes, we analyze the evolution of potentials within each phase leveraging previous work on random walks on evolving graphs. Specifically, we use the following result which is an extension of Corollary 14 in~\cite{michal}.
\begin{theorem}
\label{koucky}
(Corollary 14 in~\cite{michal}.)
After $t$ rounds of a $d_{\max}$-lazy random walk on an evolving graph with $n$ nodes, dynamic diameter $D$, upper bound on maximum degree $ d_{\max}$, and initial distribution $\vec{\Pi}_0$, the following holds.
\begin{align*}
\left|\left|\vec{\Pi}_t - \frac{\vec{I}}{n}\right|\right|_2^2 \leq \left(1-\frac{1}{ d_{\max}Dn}\right)^t\left|\left|\vec{\Pi}_0 - \frac{\vec{I}}{n}\right|\right|_2^2
\end{align*}
\end{theorem}

In between phases, \ldr nodes ``consume'' their potential, effectively changing the distribution at that point. Then, a new phase starts. 

In \nameD, given that $d$ is a function of the estimate $k$, if the estimate is low, there may be inputs for which $d$ is not an upper bound on the number of neighbors. We show in our analysis that in those cases the \ldr nodes detect the error and after some time all nodes increase the estimate.

\subsubsection*{Structure of the proof} 
The proof of correctness is structured in the following cases, depending on the relation between the size estimate $k$ and $n$. For some $\gamma$ and $\epsilon$, after completing an epoch with size estimate $k$, we prove that $\rho$, the potential accumulated by a \ldr node, must be within the following ranges.
\begin{center}
\begin{tabular}{rclc}
$k=n$ & $\Rightarrow$ & $(k-\ell) \left(1 - \frac{1}{k^\gamma}\right) \leq \rho \leq (k-\ell) \left(1 + \frac{1}{k^\gamma}\right)$ & (Lemma~\ref{manycorrect})\\
$k>n$ & $\Rightarrow$ & $\rho <  (k-\ell)\left(1-\frac{1}{k^\gamma}\right)$ & (Lemma~\ref{kaboven})\\
$k<n \leq k^{1+\epsilon}$ & $\Rightarrow$ & $\rho >  (k-\ell)\left(1+\frac{1}{k^\gamma}\right)$ & (Lemma~\ref{manyksquare})
\end{tabular}
\end{center}
For the remaining case when $k^{1+\epsilon} < n$, we prove first the following relation between $\Phi_{r+1,1}$, the potential of any node at the end of the first phase, and a threshold $\tau$. 
\begin{center}
\begin{tabular}{rclc}
$k^{1+\epsilon} < n$ & $\Rightarrow$ & $\Phi_{r+1,1} > \tau$ for at least one node & (Lemma~\ref{manyunalarmed})\\
$k\geq n$ & $\Rightarrow$ & $\Phi_{r+1,1}\leq\tau$ for all nodes & (Lemma~\ref{nolowalarm})
\end{tabular}
\end{center}
Thus, if $k^{1+\epsilon} < n$, and only if $k< n$, there is at least one node with potential above $\tau$, which moves to a status ``low'' and spreads this alarm. We complete the proof with the following.   
\begin{center}
\begin{tabular}{rclc}
$k^{1+\epsilon} < n$ & $\Rightarrow$ & all nodes receive an alarm ``low'' during phase $2$ & (Lemma~\ref{manyalarmsoon})
\end{tabular}
\end{center}


\subsubsection*{Analysis.} 
We start the analysis considering the case $k=n$ as follows.
\begin{lemma}
\label{manycorrect}
If $d \geq k = n$, 
for an ADN with $\ell < k$ \ldr nodes, 
for any $\gamma>0$ there is a $\alpha \geq \max\{2,1+\gamma+\log_k 3\}$ such that,
after running the \nameD protocol for $p \geq (2\gamma\ln k)/\left(\ell\left(\frac{1}{k}+\frac{1}{k^{\alpha}}\right)\right)$ phases, each of $r\geq 2\alpha dk^2\ln k$ rounds, the potential $\rho$ consumed by each of the $\ell$ \ldr nodes is such that 
$$(k-\ell) \left(1 - \frac{1}{k^\gamma}\right) \leq \rho \leq (k-\ell) \left(1 + \frac{1}{k^\gamma}\right).$$
\end{lemma}

\begin{proof}
Consider the vector of potentials $\vec{\Phi}_{i,1}$ at the beginning of round $1$ of any phase $i$.
We analyze the evolution of potentials within phase $i$ as a random walk on the evolving graph underlying the network topology.
Consider the initial distribution $\vec{\Pi}_1$ on the overall potential $||\vec{\Phi}_{i,1}||_1$. 
Given that $d \geq k$, for any $\alpha\geq 0$ and using Theorem~\ref{koucky}, we know that after a phase $i$ of $r\geq 2\alpha dk^2\ln k$ rounds the distribution is such that 
\begin{align}
\left|\left|\vec{\Pi}_{r+1} - \frac{\vec{I}}{k}\right|\right|_2^2 
&\leq \left(1-\frac{1}{d{\cal D}k}\right)^r\left|\left|\vec{\Pi}_1 - \frac{\vec{I}}{k}\right|\right|_2^2, \textrm{ for $d{\cal D}k > 1$, it is }\label{distanceforell}\\
&\leq \exp\left(-\frac{r}{d{\cal D}k}\right)\nonumber\\
&\leq \exp\left(-\frac{2\alpha dk^2\ln k}{d{\cal D}k}\right),\textrm{ given that $k=n>{\cal D}$, it is }\nonumber\\
&\leq \exp\left(-2\alpha\ln k\right)\nonumber\\
&= \frac{1}{k^{2\alpha}} \ . \nonumber
\end{align}
Given that $(\Pi_{r+1}[j] - 1/k)^2 \leq \left|\left|\vec{\Pi}_{r+1} - \frac{\vec{I}}{k}\right|\right|_2^2$, for any node $j$, we have that $(\Pi_{r+1}[j]-1/k)^2 \leq 1/k^{2\alpha}$ and hence 
\begin{align}
\left|\Pi_{r+1}[j] - \frac{1}{k} \right| \leq \frac{1}{k^{\alpha}} \ .\label{fractionbound}
\end{align}
Notice that the latter is true for any initial distribution.
Therefore, after each phase a \ldr node consumes 
between $1/k-1/k^{\alpha}$ and $1/k+1/k^{\alpha}$ fraction of the total
potential in the system, and the total potential in the system drops by at least
$\ell(1/k-1/k^{\alpha})$  and by at most $\ell(1/k+1/k^{\alpha})$ fraction. 
Recall that the initial overall potential in the system is $||\vec{\Phi}_{1,1}||_1=\ell(k-\ell)$.

Using the latter observations, we first find conditions on $p$ to obtain the desired bounds on $\rho$, as follows.
After $p$ phases a \ldr node consumes {\em at least}
\begin{align}
\rho &\geq \ell(k-\ell) \left(\frac{1}{k} - \frac{1}{k^{\alpha}}\right) \sum_{i=0}^{p-1} \left(1- \ell\left(\frac{1}{k}+\frac{1}{k^{\alpha}}\right)\right)^i,\label{rholb}
\end{align}
and {\em at most}
\begin{align}
\rho &\leq \ell(k-\ell) \left(\frac{1}{k} + \frac{1}{k^{\alpha}}\right) \sum_{i=0}^{p-1} \left(1- \ell\left(\frac{1}{k}-\frac{1}{k^{\alpha}}\right)\right)^i.\label{rhoub}
\end{align}
Given that $0<\ell\left(\frac{1}{k}+\frac{1}{k^{\alpha}}\right)<1$ for $\alpha \geq 2$ and $k>\ell$, Equation~\ref{rholb} is
\begin{align*}
\rho &\geq \ell(k-\ell) \left(\frac{1}{k} - \frac{1}{k^{\alpha}}\right) \frac{1-\left(1- \ell\left(\frac{1}{k}+\frac{1}{k^{\alpha}}\right)\right)^p}{1-\left(1- \ell\left(\frac{1}{k}+\frac{1}{k^{\alpha}}\right)\right)} \\
&= (k-\ell) 
\frac{ k^{\alpha}-k}{k^{\alpha}+k}
\left(1-\left(1- \ell\left(\frac{1}{k}+\frac{1}{k^{\alpha}}\right)\right)^p\right).
\end{align*}
Given that $0<\ell\left(\frac{1}{k}+\frac{1}{k^{\alpha}}\right)<1$ 
for $\ell<k$ and $\alpha \geq 2$, it is 
\begin{align*}
\rho &\geq (k-\ell) 
\frac{ k^{\alpha}-k}{k^{\alpha}+k}
\left(1-\exp\left(- p\ell\left(\frac{1}{k}+\frac{1}{k^{\alpha}}\right)\right)\right).
\end{align*}
Thus, to prove the lower bound on $\rho$, it is enough to find values of $p$ and $\alpha$ such that
\begin{align*}
\frac{ k^{\alpha}-k}{k^{\alpha}+k}
\left(1-\exp\left(- p\ell\left(\frac{1}{k}+\frac{1}{k^{\alpha}}\right)\right)\right)
&\geq 1-\frac{1}{k^\gamma} \ .
\end{align*}
We note first that for $$p\geq \frac{2\gamma\ln k}{\ell\left(\frac{1}{k}+\frac{1}{k^{\alpha}}\right)} \ ,$$ 
it is $$1-\exp\left(- p\ell\left(\frac{1}{k}+\frac{1}{k^{\alpha}}\right)\right) \geq 1-\frac{1}{k^{2\gamma}} \ .$$
Replacing, it is enough to prove
\begin{align}
\frac{ k^{\alpha}-k}{k^{\alpha}+k} \left(1+\frac{1}{k^\gamma}\right) &\geq 1 \label{eqceps}\\
k^{\alpha-\gamma} &\geq 2k + k^{1-\gamma} \ . \nonumber
\end{align}
Thus, for $\gamma>0$ it is enough to prove
$k^{\alpha-\gamma} \geq 3k$,
which is true for $\alpha \geq 1+\gamma+\log_k 3$. 
We show now the upper bound on $\rho$ starting from Equation~\ref{rhoub}:
\begin{align*}
\rho &\leq \ell(k-\ell) \left(\frac{1}{k} + \frac{1}{k^{\alpha}}\right) \sum_{i=0}^{p-1} \left(1- \ell\left(\frac{1}{k}-\frac{1}{k^{\alpha}}\right)\right)^i.
\end{align*}
Given that $1- \ell\left(\frac{1}{k}-\frac{1}{k^{\alpha}}\right)<1$ for $\alpha \geq 2>1$, it is
\begin{align*}
\rho &\leq \ell(k-\ell) \left(\frac{1}{k} + \frac{1}{k^{\alpha}}\right) \frac{1-\left(1- \ell\left(\frac{1}{k}-\frac{1}{k^{\alpha}}\right)\right)^p}{1-\left(1- \ell\left(\frac{1}{k}-\frac{1}{k^{\alpha}}\right)\right)} \\
&= (k-\ell) 
\frac{ k^{\alpha}+k}{k^{\alpha}-k}
\left(1-\left(1- \ell\left(\frac{1}{k}-\frac{1}{k^{\alpha}}\right)\right)^p\right).
\end{align*}
Given that $0<\ell\left(\frac{1}{k}-\frac{1}{k^{\alpha}}\right)<1$ and $p>0$, it is $\left(1-\left(1- \ell\left(\frac{1}{k}-\frac{1}{k^{\alpha}}\right)\right)^p\right)<1$. Then, replacing, we get
\begin{align*}
\rho &\leq (k-\ell) 
\frac{ k^{\alpha}+k}{k^{\alpha}-k} \ .
\end{align*}
Thus, to prove the upper bound on $\rho$, it is enough to find a value of $c$ such that
\begin{align*}
\frac{ k^{\alpha}+k}{k^{\alpha}-k}
&\leq 1+\frac{1}{k^\gamma} \ .
\end{align*}
This is the same as Equation~\ref{eqceps} and hence the claim follows.
\end{proof}


The previous lemma shows that, after running \nameD enough time, if 
for some \ldr node it is $ \rho > (k-\ell) \left(1 + \frac{1}{k^\gamma}\right)$
or
$\rho < (k-\ell) \left(1 - \frac{1}{k^\gamma}\right)$,
for some $\gamma>0$,
we know that the estimate $k$ is wrong. However, the complementary case, that is, 
$(k-\ell) \left(1 - \frac{1}{k^\gamma}\right) \leq \rho \leq (k-\ell) \left(1 + \frac{1}{k^\gamma}\right)$,
may occur even if the estimate is $k\neq n$ and hence the error has to be detected by other means. 
To prove correctness in that case we further separate the range of $k$ in three cases. 
The first one, when $k<n\leq k^{1+\epsilon}$, for some $\epsilon>0$, in the following lemma, 
which is based on upper bounding the potential left in the system after running \nameD long enough.
To ensure that $d\geq  \Delta+1$, we restrict $d\geq k^{1+\epsilon}$. 

\begin{lemma}
\label{manyksquare}
Under the following conditions
$1<k<n\leq k^{1+\epsilon}\leq d$, $\epsilon>0$,\\ 
after running the \nameD protocol for 
$p\geq 2\delta(\ln k)/(\ell\left(1/n+1/k^{\beta}\right))$
phases, each of 
$r\geq2\beta dk^{2+2\epsilon}\ln k$
rounds, under the following conditions
$\beta \geq \log_k (n(2k^{\delta} + 1))$,
$\beta >2$,
$\delta > \log_k (nk^\gamma/(nk^\gamma-(n-1)(k^\gamma+1)))$, and
$\gamma >\log_k (n-1)$.
Then, the potential $\rho$ consumed by any \ldr node is $\rho >  (k-\ell)\left(1+1/k^\gamma\right)$.
\end{lemma}

\begin{proof}
Given that $d\geq n$, we can use Theorem~\ref{koucky} as in Lemma~\ref{manycorrect} to show that after a phase of 
$r\geq2\beta dk^{2+2\epsilon}\ln k$ rounds the distribution is such that 
\begin{align*}
\left|\left|\vec{\Pi}_{r+1} - \frac{\vec{I}}{n}\right|\right|_2^2 
&\leq \left(1-\frac{1}{d{\cal D}n}\right)^r\left|\left|\vec{\Pi}_1 - \frac{\vec{I}}{n}\right|\right|_2^2\nonumber\\
&\leq \exp\left(-\frac{r}{d{\cal D}n}\right)\nonumber\\
&\leq \exp\left(-\frac{2\beta dk^{2+2\epsilon}\ln k}{d{\cal D}n}\right).
\end{align*}

Given that $k^{1+\epsilon}\geq n > {\cal D}$, we have that
\begin{align*}
\left|\left|\vec{\Pi}_{r+1} - \frac{\vec{I}}{n}\right|\right|_2^2 
&\leq \exp\left(-2\beta\ln k\right)\nonumber\\
&= \frac{1}{k^{2\beta}} \ .\nonumber
\end{align*}

For any node $i$, given that $(\Pi_{r+1}[i] - 1/n)^2 \leq \left|\left|\vec{\Pi}_{r+1} - \frac{\vec{I}}{n}\right|\right|_2^2$ we have that $(\Pi_{r+1}[i]-1/n)^2 \leq 1/k^{2\beta}$ and hence $\Pi_{r+1}[i] \geq 1/n - 1/k^{\beta}$.
The latter is true for any initial distribution.
Therefore, after each phase a \ldr node consumes 
at least $1/n-1/k^{\beta}$ fraction of the total
potential in the system, and the total potential in the system drops by at most $\ell(1/n+1/k^{\beta})$ fraction. 
Recall that the initial overall potential in the system is $||\vec{\Phi}_{1,1}||_1=\ell(n-\ell)$.

Using the latter observations, after $p$ phases, any given \ldr node consumes {\em at least}
\begin{align*}
\rho &\geq \ell(n-\ell) \left(\frac{1}{n} - \frac{1}{k^{\beta}}\right) \sum_{i=0}^{p-1} \left(1- \ell\left(\frac{1}{n}+\frac{1}{k^{\beta}}\right)\right)^i.
\end{align*}
Given that $0<\ell\left(\frac{1}{n}+\frac{1}{k^{\beta}}\right)<1$ for $\beta\geq 2$ and $k>\ell$, we have that
\begin{align*}
\rho &\geq \ell(n-\ell) \left(\frac{1}{n} - \frac{1}{k^{\beta}}\right) \frac{1-\left(1- \ell\left(\frac{1}{n}+\frac{1}{k^{\beta}}\right)\right)^p}{1-\left(1- \ell\left(\frac{1}{n}+\frac{1}{k^{\beta}}\right)\right)} \\
&= (n-\ell)  \frac{ k^{\beta}-n}{k^{\beta}+n}
\left(1-\left(1- \ell\left(\frac{1}{n}+\frac{1}{k^{\beta}}\right)\right)^p\right).
\end{align*}
Again using that $0<\ell\left(\frac{1}{n}+\frac{1}{k^{\beta}}\right)<1$ for $\beta\geq 2$ and $k>\ell$, 
and given that $1-x\leq e^{-x}$ for any $0<x<1$ (cf.~\cite{book:mitrinovic}),
we have that
\begin{align*}
\rho &\geq (n-\ell) \frac{ k^{\beta}-n}{k^{\beta}+n}
\left(1-\exp\left(- p\ell\left(\frac{1}{n}+\frac{1}{k^{\beta}}\right)\right)\right),
\textrm{ replacing } p\geq \frac{2\delta\ln k}{\ell\left(\frac{1}{n}+\frac{1}{k^{\beta}}\right)},\\ 
&\geq (n-\ell) \frac{ k^{\beta}-n}{k^{\beta}+n} \left(1- \frac{1}{k^{2\delta}}\right)\\
&\geq (n-\ell) \frac{ k^{\beta}-n}{k^{\beta}+n} \left(1+ \frac{1}{k^{\delta}}\right) \left(1- \frac{1}{k^{\delta}}\right)\\
&\geq (n-\ell) \left(1- \frac{1}{k^{\delta}}\right) \ .
\end{align*}
The latter inequality holds for $\beta \geq \log_k (n(2k^{\delta} + 1))$.
Then, to complete the proof, it is enough to show that
\begin{align*}
(n-\ell) \left(1- \frac{1}{k^{\delta}}\right) &> (k-\ell)\left(1+\frac{1}{k^\gamma}\right).
\end{align*}
Which is true for $k<n$, $\delta> \log_k (nk^\gamma/(nk^\gamma-(n-1)(k^\gamma+1)))$ and $\gamma>\log_k (n-1)$.
\end{proof}

We now consider the case $k^{1+\epsilon}<n$.
First, we prove the following two claims that establish properties of the potential during the execution of \nameD. 
(Recall that we use round $r+1$ to refer to potentials at the end of the phase right before \ldr nodes consume their potential in Line~\ref{leaderAlg}.\ref{phireset}.)

\begin{claim}
\label{manyconservation}
Given an ADN of $n$ nodes running \nameD with parameter $d$, for any round $t$ of the first phase, such that $1\leq t\leq r+1$, if $d$ was larger than the number of neighbors of each node $x$ for every round $t'<t$, then $||\vec{\Phi}_t||_1=(n-\ell)\ell$. 
\end{claim}

\begin{proof}
For the first round the claim holds as the initial potential of each node is $\ell$ except the \ldr nodes that get $0$. That is, $||\vec{\Phi}_1||_1 = (n-\ell)\ell$.
For any given round $1< t\leq r+1$ in phase $1$ and any given node $x$, if $d$ is larger than the number of neighbors of $x$, the potential is updated only in Lines~\ref{leaderAlg}.\ref{potupdate} and~\ref{otherAlg}.\ref{newpot} as
\begin{align*}
\Phi_{t+1}[x] &= \Phi_{t}[x] + \sum_{i\in N_{t}[x]}\Phi_{t}[i]/d - |N_{t}[x]|\Phi_{t}[x]/d
\ .
\end{align*}
Where 
$N_{t}[x]$ is the set of neighbors of node $x$ in round $t$.
Inductively, assume that the claim holds for some round $1\leq t\leq r$. 
We want to show that consequently it holds for $t+1$.
The potential for round $t+1$ is
\begin{align}
||\vec{\Phi}_{t+1}||_1 &= ||\vec{\Phi}_{t}||_1 + \frac{1}{d}\sum_{x\in V} \left( \sum_{y\in N_{t}[x]}\Phi_{t}[y] - |N_{t}[x]|\Phi_{t}[x] \right) \ .\label{manypotvecupdate}
\end{align}

In the ADN model, communication is symmetric. That is, for every pair of nodes $x,y\in V$ and round $t$, it is $x\in N_{t}[y] \iff y\in N_{t}[x]$. 
Fix a pair of nodes $x',y' \in V$ such that in round $t$ it is $y'\in N_{t}[x']$ and hence $x'\in N_{t}[y']$. 
Consider the summations in Equation~\ref{manypotvecupdate}.
Due to symmetric communication, we have that the potential $\Phi_{t}[y']$ appears with positive sign when the indices of the summations are $x=x'$ and $y=y'$, and with negative sign when the indices are $x=y'$ and $y=x'$. This observation applies to all pairs of nodes that communicate in any round $t$. 
Therefore, we can re-write Equation~\ref{manypotvecupdate} as
\begin{align*}
||\vec{\Phi}_{t+1}||_1 &= ||\vec{\Phi}_{t}||_1 + \frac{1}{d}\sum_{\substack{x,y\in V:\\y\in N_{t}[x]\\}} \bigg(\Phi_{t}[y] - \Phi_{t}[x] + \Phi_{t}[x] - \Phi_{t}[y] \bigg) \\
&= ||\vec{\Phi}_{t}||_1 \ .
\end{align*}
Thus, the claim follows.
%
\end{proof}


\begin{claim}
\label{manypotbounds}
Given an ADN of $n$ nodes running \nameD, for any round $t$ of any phase and any node $x$, it is $0\leq \Phi_t[x]\leq \ell$.  
\end{claim}

\begin{proof}
If $t=1$ the potential of the \ldr nodes is $\Phi_1[0]=0$ and the potential of any \notldr node $x$ is $\Phi_1[x]=\ell$. Thus, the claim follows. 
Inductively, for any round $2<t\leq r+1$, we consider two cases according to node status. 
If a node $x$ is in alarm status ``low'' at the beginning of round $t$, then it is $\Phi_t[x]=\ell$ because, whenever the status of a node is updated to ``low'', its potential is set to $\ell$ and will not change until the next epoch (cf. Figures~\ref{leaderAlg} and~\ref{otherAlg}).

In the second case, if a node $x$ is in ``probing'' status at the beginning of round $t$, it means that it had its potential updated in all rounds $t'<t$ only in Lines~\ref{leaderAlg}.\ref{potupdate} or~\ref{otherAlg}.\ref{newpot} as
\begin{align*}
\Phi_{t'+1}[x] &= \Phi_{t'}[x] + \sum_{y\in N_{t'}[x]}\Phi_{t'}[y]/d - |N_{t'}[x]|\Phi_{t'}[x]/d \ .
\end{align*}
For all rounds $t'<t$, node $x$ exchanged potential with less than $d$ neighbors, because otherwise it would have been changed to alarm status ``low'' in Lines~\ref{leaderAlg}.\ref{alarminsecondleader} or~\ref{otherAlg}.\ref{alarminsecondother}.
Therefore it is $|N_{t'}[x]|\Phi_{t'}[x]/d < \Phi_{t'}[x]$ which implies $\Phi_t[x]\geq 0$. 
On the other hand, it can also be seen that $\Phi_t[x]\leq \ell$ because, for any $t'<t$, it is
\begin{align*}
\Phi_{t'+1}[x] &= \Phi_{t'}[x] + \sum_{y\in N_{t'}[x]}\Phi_{t'}[y]/d - |N_{t'}[x]|\Phi_{t'}[x]/d \ .
\end{align*}
In particular for $t'=t-1$, we know by inductive hypothesis that $ \Phi_{t-1}[v] \leq \ell$ for any $v\in V$. Replacing, in the latter equation we get that $\Phi_{t}[x] \leq \ell$.
\end{proof}


To show that if $k^{1+\epsilon}<n$ \nameD detects that the estimate is low, we focus on the first phase. We define a threshold $\tau$ and a number of rounds such that, after the first phase is completed, some nodes will have potential above $\tau$ and this can happen only if the estimate is low. Then we show that \ldr nodes receive an alarm indicating that.

First, we show an upper bound of at most $k^{1+\epsilon}$ nodes with potential at most $\tau$ at the end of the first phase (Lemma~\ref{manyunalarmed}).
Thus, given that $k^{1+\epsilon}<n$, we know that there is at least one node with potential above $\tau$ at the end of the first phase.
Second, we show that if the estimate is not low, that is $k\geq n$, then all nodes have potential at most $\tau$ at the end of the first phase (Lemma~\ref{nolowalarm}).  
That is, a potential above $\tau$ can only happen when indeed the estimate is low.
Finally, we show that if $k^{1+\epsilon}<n$ an alarm ``low'' initiated by nodes with potential above $\tau$ must be received after $k^{1+\epsilon}$ further rounds of communication (Lemma~\ref{manyalarmsoon}).

\begin{lemma}
\label{manyunalarmed}
For $\epsilon>0$,
after running the first phase of the \nameD protocol, 
there are at most $k^{1+\epsilon}$ nodes that 
have potential at most $\tau=\ell(1-\ell/k^{1+\epsilon})$.
\end{lemma}

\begin{proof}
We define the \emph{slack} of node $x$ at the beginning of round $t$ as $s_t[x]=\ell-\Phi_t[x]$ and the vector of slacks at the beginning of round $t$ as $\vec{s}_t$. In words, the slack of a node is the ``room'' for additional potential up to $\ell$. 
Recall that the overall potential at the beginning of round $1$ of phase $1$ is $||\vec{\Phi}_1||_1=(n-\ell)\ell$. 
Also notice that for any round and any node $x$ the potential of $x$ is non-negative as shown in Claim~\ref{manypotbounds}.
Therefore, the overall slack with respect to the maximum potential that could be held by all the $n$ nodes at the beginning of round $1$ is $||\vec{s}_1||_1=\ell^2$.

Consider a partition of the set of nodes $\{L,H\}$, where $L$ is the set of nodes with potential at most $\tau$ at the end of the first phase, before the \ldr nodes consume their own potential in Line~\ref{leaderAlg}.\ref{consume}. That is, $\Phi_{r+1}[x] \leq \tau$ for all $x\in L$ (and $\Phi_{r+1}[y] > \tau$ for all $y\in H$).
Assume that the slack held by nodes in $L$ at the end of the first phase is at most the overall slack at the beginning of the phase. That is, $\sum_{x\in L}s_{r+1}[x] \leq ||\vec{s}_1||_1 = \ell^2$. 
By definition of $L$, we have that for each node $x\in L$ it is $s_{r+1}[x]=(\ell-\Phi_{r+1}[x])\geq \ell-\tau$.
Therefore,
$|L|(\ell-\tau) \leq \sum_{x\in L} s_{r+1}[x] \leq \ell^2$.
Thus, $|L| \leq \ell^2/(\ell-\tau) = k^{1+\epsilon}$ (because $\tau=\ell(1-\ell/k^{1+\epsilon})$) and the claim follows.

Then, to complete the proof, it remains to show that $\sum_{x\in L}s_{r+1}[x]\leq \ell^2$.
Let the scenario where $d$ is larger than the number of neighbors that each node has in each round of the first phase be called ``case 1'', and ``case 2'' otherwise.
Claim~\ref{manyconservation} shows that in case 1 at the end of the first phase it is $||\vec{\Phi}_{r+1}||_1=\ell(n-\ell)$. Therefore, the slack held by all nodes is $||\vec{s}_{r+1}||_1=\ell^2$ and the slack held by nodes in $L\subseteq V$ is $\sum_{x\in L}s_{r+1}[x]\leq \ell^2$, proving the claim for case 1. We show now that, in fact, case 1 is a worst-case scenario. That is, in the complementary case 2 where some nodes have $d$ neighbors or more in one or more rounds, the slack is even smaller. To compare both scenarios we denote the slack for each round $t$, each node $x$, and each case $i$ as $s^{(i)}_t[x]$.

Assume that some node $x$ is the first one to have $d'>d-1$ neighbors. Let $1\leq t\leq r$ be the round of the first phase when this event happened.
We claim that $||\vec{s}_{t+1}^{(2)}||_1\leq ||\vec{s}_{t+1}^{(1)}||_1$. The reason is the following.
Given that more than $d-1$ potentials are received, node $x$ increases its potential to $\ell$ for the rest of the epoch (cf. Lines~\ref{leaderAlg}.\ref{leadertoomany} and~\ref{otherAlg}.\ref{othertoomany}). That is, the slack of $x$ is $s_{t+1}^{(2)}[x]\leq s_t^{(2)}[x]=s_t^{(1)}[x]$. 
Additionally, the potential shared by $x$ with all neighbors during round $t$ is $d'\Phi_{t}[x]/d>\Phi_{t}[x](1-1/d)$ (cf. Lines~\ref{leaderAlg}.\ref{potupdate} and~\ref{otherAlg}.\ref{newpot}). That is, the potential shared by $x$ with neighbors in case 2 is more than the potential that $x$ would have shared in case 1. 
Then, combining both effects (the relative increase in potential of $x$ and its neighbors') the overall slack is $||\vec{s}_{t+1}^{(2)}||_1\leq ||\vec{s}_{t+1}^{(1)}||_1$. The same argument applies to all other nodes with $d$ or more neighbors in round $t$.

Additionally, for any round $t'$ of the first phase, such that $t<t'\leq r$, we have to also consider the case of a node $y$ that, although it does not receive more than $d-1$ potentials, it moves to alarm status ``low'' because it has received such status in round $t'$. 
Then, notice that the potential of $y$ is $\Phi_{t'+1}[y]=\ell \geq \Phi_{t'}[y]$, and it will stay in $\ell$ for the rest of the epoch (cf. Lines~\ref{leaderAlg}.\ref{alarminsecondleader} and~\ref{otherAlg}.\ref{alarminsecondother}). Therefore, the slack of $y$ is $s_{t+1}^{(2)}[y]\leq s_{t+1}^{(1)}[y]$. 

Combining all the effects studied over all rounds, the slack at the end of the first phase is $||\vec{s}_{r+1}^{(2)}||_1\leq ||\vec{s}_{r+1}^{(1)}||_1$. 
Given that $L\subseteq V$, it is $\sum_{x\in L}s_{r+1}^{(2)}[x] \leq ||\vec{s}_{r+1}^{(2)}||_1\leq ||\vec{s}_{r+1}^{(1)}||_1 \leq \ell^2$ which completes the proof.
%
\end{proof}


\begin{lemma}
\label{nolowalarm}
If $k\geq n$, $r\geq(4+2\epsilon- 2\ln(k^\epsilon-1)/\ln k)dk^2\ln k$, and $\epsilon>0$, given that $k>\ell\geq 1$, at the end of the first phase no individual node should have potential larger than $\tau=\ell(1-\ell/k^{1+\epsilon})$. 
\end{lemma}

\begin{proof}
Using Theorem~\ref{koucky}, we know that after phase $1$ of $r\geq(4+2\epsilon- 2\ln(k^\epsilon-1)/\ln k)dk^2\ln k$ rounds, the distribution is such that 
\begin{align*}
\left|\left|\vec{\Pi}_{r+1} - \frac{\vec{I}}{n}\right|\right|_2^2 
&\leq \left(1-\frac{1}{d{\cal D}n}\right)^r\left|\left|\vec{\Pi}_1 - \frac{\vec{I}}{n}\right|\right|_2^2\\
&\leq \exp\left(-\frac{r}{d{\cal D}n}\right)\nonumber\\
&\leq \exp\left(-\frac{(4+2\epsilon - 2\ln(k^\epsilon-1)/\ln k)dk^2\ln k}{d{\cal D}n}\right).
\end{align*}
Given that $k\geq n>{\cal D}$, we have that
\begin{align*}
\left|\left|\vec{\Pi}_{r+1} - \frac{\vec{I}}{k}\right|\right|_2^2 
&\leq \exp\left(-(4+2\epsilon- 2\ln(k^\epsilon-1)/\ln k)\ln k\right)\nonumber\\
&= 1/k^{4+2\epsilon- 2\ln(k^\epsilon-1)/\ln k} \ .\nonumber
\end{align*}

Given that for any node $j$, it is $(\Phi_{r+1}[j] - 1/k)^2 \leq \left|\left|\vec{\Phi}_{r+1} - \frac{\vec{I}}{k}\right|\right|_2^2$, we have that $(\Phi_{r+1}[j]-1/k)^2 \leq 1/k^{4+2\epsilon- 2\ln(k^\epsilon-1)/\ln k}$. Hence, it is  $\Phi_{r+1}[j] \leq 1/k + 1/k^{2+\epsilon- \ln(k^\epsilon-1)/\ln k}$ for any node $j$.
Moreover, if $d\geq k\geq n$ the total potential in the network would be $\ell(n-\ell)$ (cf. Claim~\ref{manyconservation}) and no individual node should have potential larger than $\ell(n-\ell)(1/k + 1/k^{2+\epsilon- \ln(k^\epsilon-1)/\ln k}) \leq  \ell(k-\ell)(1/k + 1/k^{2+\epsilon- \ln(k^\epsilon-1)/\ln k})$. We show that the latter is at most $\tau=\ell(1-\ell/k^{1+\epsilon})$ as follows:
\begin{align*}
\ell(k-\ell)\left(\frac{1}{k} + \frac{1}{k^{2+\epsilon- \ln(k^\epsilon-1)/\ln k}}\right) &\leq \ell\left(1-\frac{\ell}{k^{1+\epsilon}}\right)\\
(k-\ell)\left(\frac{1}{k} + \frac{1}{k^{2+\epsilon- \ln(k^\epsilon-1)/\ln k}}\right) &\leq 1-\frac{\ell}{k^{1+\epsilon}}\\
1-\frac{\ell}{k} + \frac{k-\ell}{k^{2+\epsilon- \ln(k^\epsilon-1)/\ln k}} &\leq 1-\frac{\ell}{k^{1+\epsilon}}\\
\frac{k-\ell}{k^{2+\epsilon- \ln(k^\epsilon-1)/\ln k}} &\leq \frac{\ell}{k}-\frac{\ell}{k^{1+\epsilon}}\\
\frac{k-\ell}{k^{2+\epsilon- \ln(k^\epsilon-1)/\ln k}} &\leq \ell\frac{k^\epsilon-1}{k^{1+\epsilon}}\\
k-\ell &\leq \ell (k^\epsilon-1)k^{1- \ln(k^\epsilon-1)/\ln k} \ , \textrm{ given that $k>\ell$,}  \\
\ln(k-\ell) &\leq \ln \ell + \ln(k^\epsilon-1) + \left(1- \frac{\ln(k^\epsilon-1)}{\ln k}\right) \ln k \\
\ln(k-\ell) &\leq \ln (\ell k) \ , \textrm{ given that $k>\ell$,}\\
k-\ell &\leq \ell k \ .
\end{align*}
And the latter is true because $\ell \geq1$.
\end{proof}



The previous lemma shows that, if the estimate is ``not-low'' ($k\geq n$), at the end of the first phase all nodes must have ``low'' potential ($\Phi_{1,r+1}\leq\tau$). 
(Notice the inverse relation between estimate and potential.)
In the following lemma we show that if $k^{1+\epsilon}<n$ (i.e. low estimate) there are some nodes with $\Phi_{1,r+1}>\tau$ (i.e. high potential), and that all the other nodes will know this within the following phase.

\begin{lemma}
\label{manyalarmsoon}
If $k^{1+\epsilon}<n$, 
$0<\epsilon\leq 1+\log_k 4d$,
and $r\geq(4+2\epsilon- 2\ln(k^\epsilon-1)/\ln k)dk^2\ln k$, 
within the following $k^{1+\epsilon}$ rounds after the first phase of the \nameD protocol, 
all \ldr nodes have received an alarm status ``low''.
\end{lemma}

\begin{proof}
Consider a partition $\{L,H\}$ of the set of nodes, 
where $L$ is the set of nodes with potential at most $\tau$ at the end of the first phase. 
As shown in Lemma~\ref{manyunalarmed} the size of $L$ is at most $k^{1+\epsilon}$, and because $k^{1+\epsilon}<n$ the size of $H$ is at least $1$.

Based on their ``high'' potential (above $\tau$),
and the property proved in Lemma~\ref{nolowalarm} (that in case of not-low estimate, $k\geq n$, there would not be any node like them),
all nodes in $H$ 
move to alarm status ``low'' at the end of phase $1$ (cf. Lines~\ref{otherAlg}.\ref{otherthreshold} or~\ref{leaderAlg}.\ref{leaderthreshold}). 
(Notice the inverse relation between potential and status, which in turn indicates whether the estimate is low or not.)
Nodes in $L$ may not have low status, but  due to $1$-interval connectivity at least one new node from $L$ moves to status ``low'' in each of the following rounds (Lines~\ref{otherAlg}.\ref{alarminsecondother} or~\ref{leaderAlg}.\ref{alarminsecondleader}).

The number of rounds of the second phase is at least 
\begin{align*}
\left(4+2\epsilon- 2\frac{\ln(k^\epsilon-1)}{\ln k}\right)dk^2\ln k 
&=  \left(4+2\epsilon- 2\log_k(k^\epsilon-1)\right)dk^2\ln k \\
&> \left(4+2\epsilon- 2\log_kk^\epsilon\right)dk^2\ln k \\
&> 4dk^2 \ .
\end{align*}
Given that the size of $L$ is at most $k^{1+\epsilon}$, to complete the proof it is enough to show that $4dk^2 \geq k^{1+\epsilon}$, which is true for $\epsilon\leq 1+\log_k 4d$.
%
%
\end{proof}



Finally, in the following lemma we show that if $k>n$, \ldr nodes detect that the potential consumed is too low for the estimate $k$ to be correct. 
\begin{lemma}
\label{kaboven}
Under the following conditions
$d>k>n>\ell>0$, 
for
$$\beta \geq \log_k (n(2k^{\delta} - 1)) \ ,$$ $$\delta > \log_k \frac{k^{\gamma}(n-\ell)}{k^\gamma-(n-\ell)-1} \ ,$$ 
$$\gamma>\log_k (n-\ell+1) \ ,$$
after running the \nameD protocol for $p$ phases and $r$ rounds such that
$$p\leq 2\delta\ln k \frac{1-\ell\left(\frac{1}{n}-\frac{1}{k^{\beta}}\right)}{\ell\left(\frac{1}{n}-\frac{1}{k^{\beta}}\right)} \ ,$$
$$r\geq2\beta dk^{2}\ln k \ ,$$
the potential $\rho$ consumed by any \ldr node is $$\rho <  (k-\ell)\left(1-\frac{1}{k^\gamma}\right) \ .$$
\end{lemma}
\begin{proof}
Given that $d>k>n$, we can use Theorem~\ref{koucky} as in Lemma~\ref{manycorrect} to show that after a phase of 
$r\geq2\beta dk^{2}\ln k$ rounds the distribution is such that 
\begin{align*}
\left|\left|\vec{\Pi}_{r+1} - \frac{\vec{I}}{n}\right|\right|_2^2 
&\leq \left(1-\frac{1}{d{\cal D}n}\right)^r\left|\left|\vec{\Pi}_1 - \frac{\vec{I}}{n}\right|\right|_2^2\nonumber\\
&\leq \exp\left(-\frac{r}{d{\cal D}n}\right)\nonumber\\
&\leq \exp\left(-\frac{2\beta dk^{2}\ln k}{d{\cal D}n}\right).
\end{align*}
Given that $k> n > {\cal D}$, we have that
\begin{align*}
\left|\left|\vec{\Pi}_{r+1} - \frac{\vec{I}}{n}\right|\right|_2^2 
&\leq \exp\left(-2\beta\ln k\right)\nonumber\\
&= \frac{1}{k^{2\beta}} \ .\nonumber
\end{align*}
For any node $i$, given that $(\Pi_{r+1}[i] - 1/n)^2 \leq \left|\left|\vec{\Pi}_{r+1} - \frac{\vec{I}}{n}\right|\right|_2^2$ we have that $(\Pi_{r+1}[i]-1/n)^2 \leq 1/k^{2\beta}$ and hence $\Pi_{r+1}[i] \geq 1/n - 1/k^{\beta}$.
The latter is true for any initial distribution.
Therefore, after each phase a \ldr node consumes 
at most $1/n+1/k^{\beta}$ fraction of the total
potential in the system, and the total potential in the system drops by at least $\ell(1/n-1/k^{\beta})$ fraction. 
Recall that the initial overall potential in the system is $||\vec{\Phi}_{1,1}||_1=\ell(n-\ell)$.
Using the latter observations, after $p$ phases, any given \ldr node consumes {\em at most}
\begin{align*}
\rho &\leq \ell(n-\ell) \left(\frac{1}{n} + \frac{1}{k^{\beta}}\right) \sum_{i=0}^{p-1} \left(1- \ell\left(\frac{1}{n}-\frac{1}{k^{\beta}}\right)\right)^i.
\end{align*}
Given that $0<\ell\left(\frac{1}{n}-\frac{1}{k^{\beta}}\right)<1$ for $\beta\geq 1$ and $k>n>\ell$, we have that
\begin{align*}
\rho &\leq \ell(n-\ell) \left(\frac{1}{n} + \frac{1}{k^{\beta}}\right) \frac{1-\left(1- \ell\left(\frac{1}{n}-\frac{1}{k^{\beta}}\right)\right)^p}{1-\left(1- \ell\left(\frac{1}{n}-\frac{1}{k^{\beta}}\right)\right)} \\
&= (n-\ell)  \frac{ k^{\beta}+n}{k^{\beta}-n}
\left(1-\left(1- \ell\left(\frac{1}{n}-\frac{1}{k^{\beta}}\right)\right)^p\right).
\end{align*}

Again using that $0<\ell\left(\frac{1}{n}-\frac{1}{k^{\beta}}\right)<1$ for $\beta\geq 1$ and $k>n>\ell$, we have that
\begin{align*}
\rho &\leq (n-\ell) \frac{ k^{\beta}+n}{k^{\beta}-n}
\left(1-\exp\left(- p\frac{\ell\left(\frac{1}{n}-\frac{1}{k^{\beta}}\right)}{1-\ell\left(\frac{1}{n}-\frac{1}{k^{\beta}}\right)}\right)\right),\\
&\hspace{2in}\textrm{ replacing } p\leq 2\delta\ln k \frac{1-\ell\left(\frac{1}{n}-\frac{1}{k^{\beta}}\right)}{\ell\left(\frac{1}{n}-\frac{1}{k^{\beta}}\right)},\\ 
&\leq (n-\ell) \frac{ k^{\beta}+n}{k^{\beta}-n} \left(1- \frac{1}{k^{2\delta}}\right)\\
&= (n-\ell) \frac{ k^{\beta}+n}{k^{\beta}-n} \left(1+ \frac{1}{k^{\delta}}\right) \left(1- \frac{1}{k^{\delta}}\right)\\
&\leq (n-\ell) \left(1+ \frac{1}{k^{\delta}}\right) \ .
\end{align*}
The latter inequality holds for $\beta \geq \log_k (n(2k^{\delta} - 1))$ and $\delta \geq \log_k (3/2)$, 
The second inequality is true because $\log_k \frac{k^{\gamma}(n-\ell)}{k^\gamma-(n-\ell)-1} > \log_k (3/2)$ for $k>n>\ell>0$. 
Then, to complete the proof, it is enough to show that
\begin{align*}
(n-\ell) \left(1+ \frac{1}{k^{\delta}}\right) &< (k-\ell) \left(1- \frac{1}{k^{\gamma}}\right),
\end{align*}
which is true for $k>n$, $\delta > \log_k \frac{k^{\gamma}(n-\ell)}{k^\gamma-(n-\ell)-1}$ and $\gamma>\log_k (n-\ell+1)$. 
Hence, the claim follows.
\end{proof}
%
%

Based on the above lemmata, we establish the correctness and running time of \nameD in the following theorem.

\begin{theorem}
\label{thm:many}
Given an ADN with $n$ nodes, which includes $\ell$ \ldr nodes such that $n>\ell\geq 1$,
and where $\ell$ is known to all nodes,
after running \nameD for each estimate $k=\ell+1,\ell+2,\ell+3,\dots,n$ with parameters: 
%
\begin{align*}
d &= k^{1+\epsilon} \ ,\\ 
p &= \left\lceil \frac{2\ln k}{\ell} \max\left\{ 
\frac{\gamma}{1/k + 1/k^{\alpha}},
\frac{\delta}{1/d+1/k^\beta}
\right\}\right\rceil \ ,\\ 
r &= \left\lceil 2dk^2(\ln k) \max\left\{ 
\alpha,
\beta k^{2\epsilon},
2+\epsilon-\frac{\ln(k^\epsilon-1)}{\ln k} 
\right\}\right\rceil \ ,\\
\tau &= \ell\left(1-\frac{\ell}{k^{1+\epsilon}}\right) \ ,
\end{align*}
under the following conditions:
%
%
\begin{align*}
\epsilon &> 0 \ ,\\
\alpha &\geq 1+\gamma+\log_k 3 \ ,\\ 
\beta &\geq \log_k(d(2k^\delta+1)) \ ,\\ 
\gamma &> \log_k (d-1) \ ,\\ 
\delta &> \log_k \frac{dk^\gamma}{k^{\gamma}+1-d} \ . 
\end{align*}
Then, all nodes stop after at most $\sum_{k\in E\cup B} (pr+d)$ rounds of communication and output $n$,
for $E=\{2^i(\ell+1):i=0,1,\dots,\log\lceil n/(\ell+1)\rceil\}$,
and $B=\{(2^{\log\lceil n/(\ell+1)\rceil}-2^i)(\ell+1):i=0,1,\dots,\log\lceil n/(\ell+1)\rceil-2\}$.
\end{theorem}

\begin{proof}
First we show that the parameters and conditions of the theorem fulfill the conditions of previous lemmas. 
Notice that for $k<n\leq k^{\epsilon}$ it is 
\begin{align*}
\gamma &> \log_k (d-1) = \log_k (k^{1+\epsilon}-1) \geq \log_k (n-1) \ ,\\
\delta 
&> \log_k \frac{dk^\gamma}{k^{\gamma}+1-d}
= \log_k \frac{k^{1+\epsilon}k^\gamma}{k^\gamma+1-k^{1+\epsilon}}\\
&\geq \log_k \frac{nk^\gamma}{k^\gamma+1-n}
= \log_k \frac{nk^\gamma}{nk^\gamma-(n-1)(k^\gamma+1)} \ ,\\
\beta &\geq \log_k(d(2k^\delta+1))
= \log_k(k^{1+\epsilon}(2k^\delta+1))
\geq \log_k(n(2k^\delta+1)) \ ,\\
\frac{\delta}{1/d+1/k^\beta}
&= \frac{\delta}{1/k^{1+\epsilon}+1/k^\beta}
\geq \frac{\delta}{1/n+1/k^\beta} \ ,
\end{align*}
as required by Lemma~\ref{manyksquare}. 
On the other hand,
for $d>k>n>\ell>0$, 
it is
\begin{align*}
\beta &\geq \log_k(d(2k^\delta+1))
\geq \log_k (n(2k^{\delta} - 1)) \ ,\\
\delta &> \log_k \frac{dk^\gamma}{k^{\gamma}+1-d}
> \log_k \frac{k^{\gamma}(n-\ell)}{k^\gamma-(n-\ell)-1} \ ,\\
\gamma &> \log_k (d-1)  
>\log_k (n-\ell+1) \ ,
\end{align*}
and
\begin{align*}
p &= \left\lceil \frac{2\ln k}{\ell} \max\left\{ 
\frac{\gamma}{1/k + 1/k^{\alpha}} \ ,
\frac{\delta}{1/d+1/k^\beta}
\right\}\right\rceil
\leq 2\delta\ln k \frac{1-\ell\left(\frac{1}{n}-\frac{1}{k^{\beta}}\right)}{\ell\left(\frac{1}{n}-\frac{1}{k^{\beta}}\right)} \ ,\\ 
r &= \left\lceil 2dk^2(\ln k) \max\left\{ 
\alpha,
\beta k^{2\epsilon},
2+\epsilon-\frac{\ln(k^\epsilon-1)}{\ln k} 
\right\}\right\rceil
\geq2\beta dk^{2}\ln k \ ,
\end{align*}
as required by Lemma~\ref{kaboven}. 
And also we verify that for $d=k^{1+\epsilon}$ it is
$1+\log_k 4d = 2+\epsilon+\log_k4 \geq \epsilon$,
as required by Lemma~\ref{manyalarmsoon}.

We prove first that \nameD is correct. To do so, it is enough to show that for each estimate $k\neq n$ the algorithm detects the error and moves to the next estimate, and that if otherwise $k=n$ the algorithm stops and outputs $k$. We consider four cases: $k=n$, $k>n$, $k<n\leq k^{1+\epsilon}$, and $k^{1+\epsilon}<n$, for some $\epsilon>0$.

Assume first that $k^{1+\epsilon}<n$. Lemmas~\ref{manyunalarmed}~and~\ref{manyalarmsoon} show that within the following $k^{1+\epsilon}$ rounds after the first phase all nodes have received and/or produced an alarm status ``low'', even if no node has detected more than $d-1$ neighbors during the execution of this phase. 
Moreover, Lemma~\ref{nolowalarm} shows that the event triggering this alarm would not have happened if the estimate was not low.
For the given function $p$ and $k\geq \ell+1$, the epoch has more than one phase. Therefore, within $k^{1+\epsilon}$ rounds into the second phase, all nodes will have status ``low'' (Lines~\ref{leaderAlg}.\ref{alarminsecondleader} and~\ref{otherAlg}.\ref{alarminsecondother}), and will not change their status in this epoch (which could only happen if \ldr nodes were in status ``probing'' in Line~\ref{leaderAlg}.\ref{statuscheck}). Hence, all nodes will continue to the next epoch after updating $k$ accordingly in Lines~\ref{leaderAlg}.\ref{leaderupdate} and~\ref{otherAlg}.\ref{otherupdate}. 

Assume now that $k<n\leq k^{1+\epsilon}$. Even if some \ldr nodes do not receive an alarm status ``low'' during the execution, as shown in Lemma~\ref{manyksquare}, at the end of the epoch in Line~\ref{leaderAlg}.\ref{toolow} all \ldr nodes will detect that $\rho$ is too large and will change their status to ``low''. Then, they will disseminate their status in the loop of Line~\ref{leaderAlg}.\ref{leadernotification}, which has enough iterations to reach all nodes due to $1$-interval connectivity and $d=k^{1+\epsilon} \geq n$.
Then, all nodes will continue to the next epoch after updating $k$ appropriately.

If $k>n$, Lemma~\ref{kaboven}~~shows that the accumulated potential $\rho$ of every \ldr node after running \nameD for $p$ phases will be less than $(k-\ell)(1-1/k^\gamma)$, which is detected by the \ldr nodes in Line~\ref{leaderAlg}.\ref{toobig}. Thus, all \ldr nodes will change their status to ``high''.
Then, they will disseminate their status in the loop of Line~\ref{leaderAlg}.\ref{leadernotification}, which has enough iterations to reach all nodes  due to $1$-interval connectivity and $d>k>n$. Then, all nodes will continue to the next epoch after updating $k$ appropriately.

Finally, if $k=n$, Lemma~\ref{manycorrect} shows that the accumulated potential $\rho$ of every \ldr node after running \nameD long enough will be $(k-\ell)(1-1/k^\gamma) \leq \rho \leq (k-\ell)(1+1/k^\gamma)$. Thus, in Line~\ref{leaderAlg}.\ref{range} all \ldr nodes will change their status to ``done'', and in the loop of Line~\ref{leaderAlg}.\ref{leadernotification} will inform all other nodes that the current estimate is correct disseminating their status. 
The number of iterations is enough to do so due to $1$-interval connectivity.

The claimed running time can be obtained by inspection of either algorithm, \ldr nodes or \notldr nodes, since they are synchronized.
Refer for instance to Algorithm~\ref{leaderAlg}.
For each epoch, Line~\ref{leaderAlg}.\ref{phasesleader} starts a loop of $p$ phases followed by $d$ rounds in Line~\ref{leaderAlg}.\ref{leadernotification}. 
Each of the $p$ phases has $r$ rounds. 
In each round, nodes communicate exactly once.
Then, the number of communication rounds on each epoch is $pr+d$.
The parameters $p$, $r$, and $d$ are all monotonically increasing functions of $k$. 
Hence, the running time $pr+d$ for each epoch depends on the estimate size $k$ used in that epoch, and if the values of $k$ for a given network depend on the execution, the execution with largest values of $k$ is the worst case. 

For a given network, the sequence of values of $k$ used is initially $2^{i}(\ell+1)$ for $i=0,1,\dots,\log\lceil n/(\ell+1)\rceil$, which is the set $E$ defined in the theorem, independently of the execution. 
Once $i=\log\lceil n/(\ell+1)\rceil$, if $k\neq n$, $k$ will be updated in a binary search fashion within the range $(2^{i-1}(\ell+1),2^{i}(\ell+1))$. The worst case then corresponds to a binary search concentrated in the larger half of each successive range. Such sequence of values of $k$ is $(2^{\log\lceil n/(\ell+1)\rceil}-2^i)(\ell+1)$ for $i=0,1,\dots,\log\lceil n/(\ell+1)\rceil-2$, which is the set $B$ defined in the theorem. 

Thus, the claimed worst-case number of communication rounds follows.
\end{proof}

\begin{corollary}
\label{cor:mmctime}
The time complexity of \nameD on an ADN with $\ell$ \ldr nodes and $n-\ell$ \notldr nodes is $O\left(\frac{n^{4+\epsilon}}{\ell} \log^3 n \right)$, for any $\epsilon>0$.
\end{corollary}

\begin{proof}
The following algebraic manipulations refer to the parameters and conditions in Theorem~\ref{thm:many}.
Fixing 
$\alpha = 1+\gamma+\log_k 3$,
$\beta = \log_k(d(2k^\delta+1))$,
and $\gamma = \log_k d$,
the conditions on $\alpha$, $\beta$, and $\gamma$ are fulfilled. 
Replacing $d=k^{1+\epsilon}$ and $\gamma = \log_k d$ in the condition on $\delta$, we have that
$\delta > 2(1+\epsilon)$.
Therefore, in the maximization in the definition of $p$, we get that
$\max\left\{\frac{\delta}{1/d+1/k^\beta} , \frac{\gamma}{1/k + 1/k^{\alpha}} \right\} =
\frac{\delta}{1/d+1/k^\beta}$.
Thus, it is
\begin{align*}
p &= \left\lceil \frac{2\ln k}{\ell} \frac{\delta}{1/d+1/k^\beta}\right\rceil \\
&= \left\lceil \frac{2\ln k}{\ell} \frac{\delta}{1/d+1/(d(2k^\delta+1))}\right\rceil \\
&= \left\lceil \frac{2d\ln k}{\ell} \frac{\delta}{1+1/(2k^\delta+1)}\right\rceil \\
&< \left\lceil \frac{2d\ln k^{\delta}}{\ell} \right\rceil.
\end{align*}

On the other hand, replacing $\alpha$, $\beta$ and $\gamma$ in the maximization in the definition of $r$, that is
$\max\left\{\alpha, \beta k^{2\epsilon} , 2+\epsilon-\frac{\ln(k^\epsilon-1)}{\ln k}\right\}$, 
we observe that 
$\beta k^{2\epsilon} > 2+\epsilon-\frac{\ln(k^\epsilon-1)}{\ln k}$
and
$\beta k^{2\epsilon} > \alpha$.
Thus, it is 
\begin{align*}
r &= \left\lceil 2dk^2(\ln k) \beta k^{2\epsilon}\right\rceil
\\&= \left\lceil 2d^3\beta \ln k \right\rceil\\
&= \left\lceil 2d^3  \log_k(d(2k^\delta+1)) \ln k \right\rceil\\
&= \left\lceil 2d^3  \ln(d(2k^\delta+1)) \right\rceil.
\end{align*}

Then, replacing in $pr+d$, we get 
\begin{align*}
pr+d 
< 
 \left\lceil \frac{2d}{\ell} \ln k^{\delta}\right\rceil
\left\lceil 2d^3  \ln(d(2k^\delta+1)) \right\rceil 
+ d
\in O\left(\frac{d^4}{\ell} \log^2 d\right).
\end{align*}

Replacing $d=k^{1+\epsilon} < (2n)^{1+\epsilon}$ and noticing that 
the total number of terms in the summation of the running time of Theorem~\ref{thm:many} is $O(\log (n/\ell))$,
the claim follows.
\end{proof}



\section{Randomized Unconscious Counting}
\label{sec:randomized}

\newcommand{\Count}{Count}
\newcommand{\EmptyThreads}{EmptyThreads}

In this section, we present and analyze our randomized algorithm for Unconscious Counting~\cite{conscious} as defined in Section~\ref{prelim}, that is, in finite time all nodes must obtain the correct count, possibly not knowing 
when exactly they have got~it.

In our algorithm, that we call  \LLMC (\nameR), nodes update their count, possibly multiple times.
We show that there is a point in the execution where all nodes have obtained the correct count, and do not change it to an incorrect count anymore. 
Nodes do not know when the count is final, hence they continue executing the algorithm forever. 
Although \nameR does not terminate, our analysis shows an upper bound on the time to obtain the final count. Given that the time bound is a function of $n$ and $n$ is unknown, it cannot be used as a termination condition. 

We first give the intuition of the algorithm referring to the pseudocode in Algorithm~\ref{randAlg}. \nameR uses a modified version of \nameD as a subroutine. We summarize those modifications in Algorithm~\ref{MMCAlgModified}. Further pseudocode details can be found in Algorithms~\ref{leaderAlgModif} and~\ref{otherAlgModif}. 
Recall that references to algorithms lines are given as $\langle algorithm\#\rangle.\langle line\#\rangle$.

The main idea of \nameR is the following. 
We consider consecutive powers of $2$ as values of some variable $K$. 
For each such $K$, we select \ldr nodes locally with probability corresponding to the inverse of $K$,
and we run \nameD configured for one \ldr node. 
We call this modified algorithm \nameT with parameters $K$ and $\ellp$, where $\ellp$ is the aimed number of \ldr nodes. Thus, it is $\ellp=1$, whereas $\ell$ is the actual number of \ldr nodes after the random choices.

If $K\geq n$ and there is indeed one \ldr node (i.e. $\ell=\ellp=1$),
\nameT guarantees that all nodes obtain the exact count $n$ of nodes.  
Because $K$ is doubled iteratively at some point it will be $K\geq n$. 
There is however a problem: what to do when $K\geq n$ but there is no \ldr node or at least two \ldr nodes.
Algorithm \nameR overcomes this issue with two techniques:
\begin{itemize}
\item 
Introducing parallel threads (Line~\ref{randAlg}.\ref{threads}) and carefully counting locally
the number of threads with no \ldr nodes recorded (Line~\ref{randAlg}.\ref{counthreads})
and requiring their ratio to be bigger than half (Line~\ref{randAlg}.\ref{llmciter}).
The latter is to 
have guarantees that in at least one thread it is $\ell=1$ with sufficiently high probability.
\item
Making use of the fact
that having more than one \ldr node in the execution of \nameT
configured for $\ellp=1$ (as a subroutine of \nameR) cannot return
an estimate bigger than if it was run with one \ldr node;
therefore, taking the maximum of returned estimates over
threads (Line~\ref{randAlg}.\ref{maxcount}) addresses the potential problem of more than
one \ldr node.
\end{itemize}

The main control parameter used in \nameR is $K$, which is an upper bound for estimates considered in one execution of the while loop (Line~\ref{randAlg}.\ref{llmciter}), which we call \emph{iteration} $K$. 
We start from a value of $K$ (Line~\ref{randAlg}.\ref{startK}) appropriate for our analysis. 
Within one iteration, we initiate $f(K)$ parallel threads (Line~\ref{randAlg}.\ref{threads}).
That is, nodes run $f(K)$ instances of the same algorithm specified in Lines~\ref{randAlg}.\ref{beginthread} to~\ref{randAlg}.\ref{endthread}, sharing among all instances a variable $EmptyThreads$ and a set $Count$ for each node.
For each thread, we select \ldr nodes for the whole thread (Line~\ref{randAlg}.\ref{selblacks}) ---  trying to make sure that the chance of getting one \ldr node is sufficiently large,
especially for $K\geq n$ (recall that we could not recognize for sure whether $K\ge n$ or not).
Then, in each thread independently, we run \nameT configured for $\ellp=1$ as a subroutine (Line~\ref{randAlg}.\ref{MMCcall}), hoping that we selected exactly one \ldr node ($\ell=\ellp$) in the beginning of the iteration. \nameT checks all possible values of $k$ from $1$ to $K$ aiming to find a good estimate of $n$. I.e., in \nameT we trim the execution of \nameD to estimates $k\le K$ (Line~\ref{MMCAlgModified}.\ref{trim}). 

This approach does not work for $K<n$, as then the probability
of getting more than one \ldr node could be bigger than the one for one \ldr node; however, we can eliminate such cases by monitoring
the number of threads with no \ldr node (Line~\ref{randAlg}.\ref{counthreads}), which in case of
$K<n$ should be compared with a large threshold (Line~\ref{randAlg}.\ref{llmciter}). 
Intuitively, we want \nameR to reach this threshold with high probability when $K$ is close to $n$.
Note that in \nameT a no-\ldr-node thread will output an indicator of such event, and a thread that does not reach a result (e.g. if $\ell=1$ but $K<n$) will output a count of zero (Line~\ref{MMCAlgModified}.\ref{returncount}).
If instead a thread identifies a good estimate (i.e. \nameT returns a count larger than $0$ in that thread) it adds its value to the set $\Count$ and rather than entering the next iteration, stops returning the maximum value in the set $\Count$ (Line~\ref{randAlg}.\ref{maxcount}). 
The latter is to give
preference to threads with one \ldr node over those with more than
one \ldr node (we will argue that threads with more than one \ldr node could return values, but not bigger than ones by threads with one \ldr node).

For $K<n$, there is a chance of having some threads with undetected \ldr nodes for some nodes and detected for others, and some of those nodes may obtain a count, but such count will be smaller than $n$ because $K<n$ and \nameT outputs an estimate $k\leq K$. Thus, this incorrect count will be updated later on when $K$ increases (cf. Line~\ref{randAlg}.\ref{maxcount}).

For $K\geq n$ there will be an iteration when all nodes obtain the correct count with large enough probability, as we show in our analysis. In later iterations, some threads may have more than one \ldr node yielding \nameT to return an incorrect count, but such count will be smaller than $n$. Hence, nodes disregard it (cf. Line~\ref{randAlg}.\ref{maxcount}).

\begin{algorithm}[htbp]
\caption{\nameR algorithm for each node. $\zeta \in (0,1)$ and $K$ is notation for the maximum size estimate.}
\label{randAlg}
\DontPrintSemicolon
$count \gets 0$\;
$K \gets \lceil\lceil 12/\zeta \rceil\rceil$ \label{initK} \tcp*{$\lceil\lceil x \rceil\rceil$: the smallest power of 2 bigger than $x$} \label{startK}
$\Count \gets \emptyset$ \tcp*{set of potentially "good" estimates computed} 
$\EmptyThreads\gets 0$ \tcp*{\# threads with no \ldr node detected}
\While{true}{ \label{llmciter}
	$\Count \gets \emptyset$, $\EmptyThreads\gets 0$ \;
	$K \gets 2K$ \;
	Initiate $f(K) = 64\frac{\log (K/\zeta)}{\log(e/(e-2))}$ parallel threads \label{threads} \tcp*{parallel computation and messages sharing same resources/medium}
	\For{each thread}{ 
		\For{each node}{ \label{beginthread}
			Select to be a \ldr node with probability $1/g(K)$, where $g(K)=K/2$ \; \label{selblacks}
		}
		$\langle k,b\rangle \gets \nameT(K,1)$ \label{MMCcall} \tcp*{refer to Algorithm~\ref{MMCAlgModified}}
		\lIf{$k>0$}{$\Count \gets \Count\cup\{k\}$\label{store}}
		\If(\tcp*[f]{no \ldr node detected}){$b = false$}{ 
			$\EmptyThreads \gets \EmptyThreads +1$ \; \label{counthreads}
		}\label{endthread}
	}
	\If{$\Count \neq \emptyset$ and $\EmptyThreads > f(K)/2$}{
		$count \gets \max\{count,\max(\Count)\}$ \tcp*{update $n$} \label{maxcount}
	}
}

\end{algorithm}

\begin{algorithm}[htbp]
\caption{Summary of modifications of \nameD to be used as a subroutine for \nameR. Refer to Algorithms~\ref{leaderAlgModif} and~\ref{otherAlgModif} for the complete details of \nameT. $\ellp$ is the aimed number of \ldr nodes and $K$ is the maximum size estimate.}
\label{MMCAlgModified}
\DontPrintSemicolon
\SetKwFunction{KwFn}{\nameT}
\SetKwProg{Fn}{Function}{}{end}
\Fn{\KwFn{$K$,$\ellp$}}{
Run \nameD modified as follows: \;
\hspace{.2in} -- Stop iterations when size estimate $k>K$ \label{trim} \label{stopbigk} \;
\tcp*{Lines~\ref{leaderAlgModif}.\ref{leaderitercond} and~\ref{otherAlgModif}.\ref{otheritercond}}
\hspace{.2in} -- If estimate $k<K$, remain idle until end of phase $K$ \;
\tcp*{for synch, Lines~\ref{leaderAlgModif}.\ref{leadersynchdelay}and~\ref{otherAlgModif}.\ref{othersynchdelay}}

\hspace{.2in} -- Include a Boolean $b_j$ in each node $j$ as follows: \; \label{initiatebegins}
\hspace{.4in} --- Initially: \;
\hspace{.6in} \leIf{node $j$ is \ldr}{$b_j\gets true$}{$b_j\gets false$} \tcp*{Lines~\ref{leaderAlgModif}.\ref{leaderinitp} and~\ref{otherAlgModif}.\ref{otherinitp}}
\hspace{.4in} --- In each iteration: \;
\hspace{.6in} Broadcast and Receive messages including $b_j$ \; \label{blackflagbroadcast}
\hspace{.6in} \lIf{$b_i=true$ received from some neighbor $i$}{$b_j\gets true$} \label{initiateends} \tcp*{Lines~\ref{otherAlgModif}.\ref{otherupdatep1} and~\ref{otherAlgModif}.\ref{otherupdatep2}}

Upon completion, for each node $j$: \; 
\hspace{.2in} \leIf{$status = done$}{ \textbf{return} $\langle k,b_j\rangle$}{\textbf{return} $\langle 0,b_j\rangle$} \label{returncount}  \label{zerocount} \tcp*{Lines~\ref{leaderAlgModif}.\ref{leaderreturn} and~\ref{otherAlgModif}.\ref{otherreturn}}
}
\end{algorithm}

\begin{algorithm}[htbp]
\caption{\nameT algorithm for each {\bf\emph{\ldr node}}. 
Sections that are identical to \nameD are grayed-out.
$N$ is the set of neighbors of this node in the current round, 
$\ellp$ and $K$ are notations for the aimed
number of \ldr nodes 
and the maximum size estimate respectively.
The parameters $d,p,r$ and $\tau$ are as defined in Theorem~\ref{thm:many} for \nameD, as functions of $k,\ellp$ instead of $k,\ell$, 
and $round_{\max} = \sum_{k=2^i:i=1,2,\dots,\lceil\log K\rceil} (pr+d)$.}
\label{leaderAlgModif}
\DontPrintSemicolon
\SetKwFunction{KwFn}{\nameT}
\SetKwProg{Fn}{Function}{}{end}
\Fn{\KwFn{$K$,$\ellp$}}{

\textcolor{gray}{
	$k \gets \ellp+1, min\gets k, max\gets\infty$ \tcp*{initial size estimate and range}
}
	$b \gets true$ \label{leaderinitp} \tcp*{existence of \ldr nodes}
	$\#rounds \gets 0$ \tcp*{counter of communication rounds for synchronization}
	\Repeat(\tcp*[f]{iterating epochs}){$status=done$ {\bf or} $k>K$}{  \label{epochsleaderModif}
\textcolor{gray}{
		 $status\gets probing$, $\Phi\gets 0$ \tcp*{status$=$probing$|$low$|$high$|$done, current potential}
		 $\rho\gets 0$ \tcp*{potential accumulator}
		\For(\tcp*[f]{iterating phases}){$phase=1$ to $p$}{  \label{phasesleaderModif}
			\For(\tcp*[f]{iterating rounds}){$round=1$ to $r$}{  \label{roundsleaderModif}
				 Broadcast $\langle\Phi,status,b\rangle$
				and Receive $\langle\Phi_i,status_i,b_i\rangle, \forall i\in N$ \;
				$\#rounds \gets \#rounds +1$ \;
				\If{$status=probing$ {\bf and} $|N|\leq d-1$ {\bf and} $\forall i\in N:status_i=probing$}{ 
					 $\Phi\gets \Phi + \sum_{i\in N}\Phi_i/d - |N|\Phi/d$ \tcp*{update potential} \label{potupdateModif}	
				}
				\Else(\tcp*[f]{$k<n$}){ \label{leadertoomanyModif}
					 $status\gets low$, \label{alarminsecondleaderModif}
					$\Phi\gets \ellp$
				}
			} 
			\If(\tcp*[f]{$k<n$} ){$phase=1$ {\bf and} $\Phi> \tau$}{ 
					 $status\gets low$,
					$\Phi\gets \ellp$ \label{leaderthresholdModif}
			} 
			\If(\tcp*[f]{prepare for next phase}){$status=probing$}	{ \label{rhoupdateModif}
				 $\rho \gets \rho + \Phi$ \label{consumeModif} \tcp*{consume potential} 
				 $\Phi \gets 0$ \label{phiresetModif}
			}
		} 
		\If{$status=probing$}{ \label{statuscheckModif}
			\lIf(\tcp*[f]{$k=n$}){$(k-\ellp)(1-k^{-\gamma})\leq \rho \leq (k-\ellp)(1+k^{-\gamma})$}{ 
				$status\gets done$  \label{rangeModif}
			}
			\lIf(\tcp*[f]{$k>n$}){$\rho < (k-\ellp)(1-k^{-\gamma})$}{ 
				$status\gets high$  \label{toobigModif}
			}
			\lIf(\tcp*[f]{$k<n$}){$\rho > (k-\ellp)(1+k^{-\gamma})$}{ 
				$status\gets low$  \label{toolowModif}
			}
		}
}
		\For(\tcp*[f]{disseminate status and existence of \ldr nodes}){$round=1$ to $d$}{  
			 Broadcast $\langle status, b\rangle$
			and Receive $\langle status_i, b_i\rangle, \forall i\in N$ \;
				$\#rounds \gets \#rounds +1$ \;
		} 
\textcolor{gray}{
		\If(\tcp*[f]{prepare for next epoch}){$status=low$}{  \label{leaderupdateModif}
			$min\gets k+1$ \;
			\leIf{$max=\infty$}{$k\gets 2k$}{$k\gets\lfloor(min+max)/2\rfloor$}
		}
		\Else{ 
			\If{$status=high$}{
		         	 $max\gets k-1$\;
				 $k\gets\lfloor(min+max)/2\rfloor$\;
			}
		}
}
	} \label{leaderitercond}
	\While(\tcp*[f]{for synchronization among threads} ){$\#rounds \leq round_{\max}$}{ \label{leadersynchdelay}
			 Broadcast $\langle status, b\rangle$
			and Receive $\langle status_i, b_i\rangle, \forall i\in N$ \;
			 $\#rounds \gets \#rounds +1$ \;
	}
	\leIf{$status=done$}{\textbf{return} $\langle k,b\rangle$}{\textbf{return} $\langle0,b\rangle$} \label{leaderreturn}

}
\end{algorithm}

\begin{algorithm}[htbp]
\caption{\nameT algorithm for each {\bf\emph{\notldr node}}.
Sections that are identical to \nameD are grayed-out. 
$N$ is the set of neighbors of this node in the current round, 
$\ellp$ and $K$ are notations for the aimed
number of \ldr nodes 
and the maximum size estimate respectively.
The parameters $d,p,r$ and $\tau$ are as defined in Theorem~\ref{thm:many} for \nameD, but as functions of $k,\ellp$ instead of $k,\ell$, 
and $round_{\max} = \sum_{k=2^i:i=1,2,\dots,\lceil\log K\rceil} (pr+d)$.}
\label{otherAlgModif}
\DontPrintSemicolon
\SetKwFunction{KwFn}{\nameD}
\SetKwProg{Fn}{Function}{}{end}
\Fn{\KwFn{$K$,$\ellp$}}{

\textcolor{gray}{
	 $k \gets \ellp+1, min\gets k, max\gets\infty$ \tcp*{initial size estimate and range}
}
	 $b \gets false$ \label{otherinitp} \tcp*{existence of \ldr nodes}
	 $\#rounds \gets 0$ \tcp*{counter of communication rounds for synchronization}
	\Repeat (\tcp*[f]{iterating epochs}){$status=done$ {\bf or} $k>K$}{ \label{epochsotherModif}
\textcolor{gray}{
		 $status\gets probing$, $\Phi\gets \ellp$ \tcp*{status$=$probing$|$low$|$high$|$done, current potential}
		\For( \tcp*[f]{iterating phases}){$phase=1$ to $p$}{ \label{phasesotherModif}
			\For(\tcp*[f]{iterating rounds}){$round=1$ to $r$}{  \label{roundsotherModif}
				 Broadcast $\langle\Phi,status,b\rangle$
				and Receive $\langle\Phi_i,status_i,b_i\rangle, \forall i\in N$ \;
				$\#rounds \gets \#rounds +1$ \;
				\lIf{$\exists i\in N : b_i = true$}{$b\gets true$} \label{otherupdatep1}
				\If{$status=probing$ {\bf and} $|N|\leq d-1$ {\bf and} $\forall i\in N:status_i=probing$}{ 
					 $\Phi\gets \Phi + \sum_{i\in N}\Phi_i/d - |N|\Phi/d$ \label{newpotModif}
					\tcp*[f]{update potential}	
				}
				\lElse (\tcp*[f]{$k<n$}){ \label{othertoomanyModif}
					 $status\gets low$, \label{alarminsecondotherModif}
					$\Phi\gets \ellp$
				}
			}
			\lIf( \tcp*[f]{$k<n$} ){$phase=1$ {\bf and} $\Phi> \tau$}{
					 $status\gets low$,
					$\Phi\gets \ellp$ \label{otherthresholdModif}
			} 
		}
}
		\For( \tcp*[f]{disseminate status} ){$round=1$ to $d$}{
			 Broadcast $\langle status, b\rangle$
			and Receive $\langle status_i, b_i\rangle, \forall i\in N$ \;
			\lIf{$\exists i\in N:status_i\neq probing$}{$status\gets status_i$}
			$\#rounds \gets \#rounds +1$\;
			\lIf{$\exists i\in N : b_i = true$}{$b\gets true$} \label{otherupdatep2}
		}
\textcolor{gray}{
		\If(\tcp*[f]{prepare for next epoch}){$status=low$}{  \label{otherupdateModif}
			 $min\gets k+1$\;
			\leIf{$max=\infty$}{$k\gets 2k$}{$k\gets\lfloor(min+max)/2\rfloor$}
		}
		\Else{ 
			\If{$status=high$}{
		         	 $max\gets k-1$\;
				 $k\gets\lfloor(min+max)/2\rfloor$\;
			}
		}
}
	} \label{otheritercond}
	\While(\tcp*[f]{for synchronization among threads}){$\#rounds \leq round_{\max}$}{  \label{othersynchdelay}
			 Broadcast $\langle status, b\rangle$
			and Receive $\langle status_i, b_i\rangle, \forall i\in N$\;
			 $\#rounds \gets \#rounds +1$\;
			\lIf{$\exists i\in N : b_i = true$}{$b\gets true$}
	}
	\leIf{$status=done$}{\textbf{return} $\langle k,b\rangle$}{\textbf{return} $\langle0,b\rangle$} \label{otherreturn}

}
\end{algorithm}


\subsection{Analysis of \nameR}

Throughout the analysis that follows, we consider an adaptive adversary, 
as we allow network changes to be done online
when viewing the whole history of the computation
up to the current round.

Recall that we refer to an execution of the while loop in Line~\ref{randAlg}.\ref{llmciter} for fixed parameter $K$ as iteration $K$.
Recall also that $K$
is a power of $2$ bigger than $\lceil\lceil12/\zeta\rceil\rceil$ (cf. Line~\ref{randAlg}.\ref{initK}). 
By $\lceil\lceil x 
\rceil\rceil$ we denote the smallest power of $2$ bigger than~$x$.

We conjecture that all arbitrary constants in the algorithm,
i.e., in the definition of the starting value of $K$, functions
$f(K)$ and $g(K)$, could be substantially lowered, as we
set them high to avoid too many cases in the analysis
(so making it as much focused on main arguments as possible).

In the analysis of \nameR that follows, we refer to lemmas proving properties of \nameD. Those properties also hold for \nameT, after changing $\ell$ by $\ellp$ and $n$ by $K$ as needed. 

While $K<n$, nodes may obtain a count, and such count may be incorrect. 
However,
given that \nameR continues executing doubling $K$, we focus 
on: showing that for $K\geq n$ the count obtained is correct, 
bounding the time to reach that state,
and 
showing that after the correct count is obtained the nodes do not go back to an incorrect count. 
We start from stating the following structural property.

\begin{proposition}
\label{f:max}
For any iteration $K\ge n$, 
if the number of threads with no \ldr node is bigger than $f(K)/2$, and there is at least one thread with one \ldr node,
then all threads with one \ldr node store the same value in the set $\Count$
and it is the biggest value stored in $\Count$ in this iteration.
\end{proposition}
\begin{proof}
By properties of algorithm \nameT run for $\ellp=1$ as a subroutine, when $K\geq n$, threads with one \ldr node store the correct value of $n$ in the set $\Count$ (c.f., Theorem~\ref{thm:many} for \nameD and Line~\ref{randAlg}.\ref{store}).
Threads with no \ldr node do not store any value in $\Count$, while in threads with at least two \ldr nodes, \ldr nodes acummulate a potential that yield a count at most $n$, hence nodes store a value at most $n$ in $\Count$ (c.f. {Line~\ref{leaderAlgModif}.\ref{rangeModif} in \nameT or Line~\ref{leaderAlg}.\ref{range} in \nameD}).
Therefore, by taking the maximum value from the set $\Count$ in Line~\ref{randAlg}.\ref{maxcount} in \nameD, each node obtains the correct count in the considered case.
\end{proof}

\begin{lemma}
	\label{lem:good-stop}
Consider an iteration $K\ge n$. 
The probability of event: in iteration $K$ all nodes 
obtain the correct count, is at least
$1-\exp(-f(K)/64) - \exp\left(-\frac{nf(K)}{g(K)e}\right)$.
\end{lemma}
\begin{proof}
Observe that the considered event is implied by the following
case: the number of threads without a \ldr node is 
bigger than $f(K)/2$ and there is a thread with exactly one \ldr node.
Thus, by Proposition~\ref{f:max} all nodes obtain the correct value.
Therefore, it is enough to estimate from below the probability 
of the case to hold. 

The probability of a fixed thread to be without a \ldr node
is 
\[
\left(1-\frac{1}{g(K)}\right)^{n}
\ge
\min\left\{4^{-n/g(K)},1-\frac{n}{g(K)}\right\} 
\ge
\frac{3}{4}
\]
for $K\ge 8n$,
by simple calculus and definition of function $g$. 
Therefore, the expected number of threads with no \ldr node
is at least $3f(K)/4$ for $K\ge 8n$. By Chernoff bound, the probability that this number is bigger than $f(K)/2$ is
at least $1-\exp(-f(K)/64)$, for $K\ge 8n$.

On the other hand, the probability that there is a thread with exactly one \ldr node is
\begin{align*}
1-\left( 1-\frac{n}{g(K)}\left( 1-\frac{1}{g(K)}\right)^{n-1}\right)^{f(K)}
&\ge
1-\left(1-\frac{n}{g(K)e}\right)^{f(K)}\\
&\ge
1-\exp\left(-\frac{nf(K)}{g(K)e}\right)
\ .
\end{align*}

Putting the two probabilities together, we get that
the probability of the considered event is at least
\[
1-\exp(-f(K)/64) - \exp\left(-\frac{nf(K)}{g(K)e}\right)
\ .
\]
\end{proof}

\begin{corollary}
\label{cor:good-stop}
Consider iteration $K= \lceil\lceil 8n \rceil\rceil$. 
The probability of event: in iteration $K$ all nodes 
obtain the correct count, is at least
$1-\zeta/n$.
\end{corollary}

\begin{proof}
By substituting $K$ by $\lceil\lceil 8n \rceil\rceil$ in
Lemma~\ref{lem:good-stop} we get the probability of the
considered event to be at least 
$1-\exp(-f(K)/64) - \exp\left(-\frac{nf(K)}{g(K)e}\right)
\ge 
1 - \zeta/(2n) - \zeta/(2n)
= 
1-\zeta/n$.
\end{proof}

\begin{theorem}
	\label{thm:rand-alg}
For any given $\epsilon>0$ and $\zeta>0$, 
with probability at least $1-\zeta$, 
all nodes running \nameR 
obtain the correct count $n$ 
in $O((n+1/\zeta)^{4+\epsilon}\log^3 (n+1/\zeta))$ rounds and will not change it to an incorrect count after that, 
even when running against an adaptive adversary.
\end{theorem}

\begin{proof}
By Corollary~\ref{cor:good-stop}, 
all nodes 
obtain the correct count at the end of iteration 
$K=\lceil\lceil 8n \rceil\rceil$, with probability at least $1-\zeta/n$.

As $K$ increases beyond $\lceil\lceil 8n \rceil\rceil$, the probability of having more than $f(K)/2$ threads without \ldr nodes increases. Thus, threads with \ldr nodes may produce a count in future iterations. If there is a thread with exactly one \ldr node, all nodes obtain the correct count by Proposition~\ref{f:max}. If on the other hand in those threads either $\ell=0$ or $\ell>1$, nodes do not obtain a new count or obtain a count smaller than $n$ respectively. Hence, nodes do not update their count in Line~\ref{randAlg}.\ref{maxcount}.



The total number of communication rounds for each iteration of \nameR is the number of rounds defined in \nameT, which is $round_{\max} = \sum_{k=2^i:i=1,2,\dots,\lceil\log K\rceil} (pr+d)$ (recall that $p$, $r$, and $d$ depend on $k$).
The value of $round_{\max}$ corresponds to the worst case scenario (with respect to running time) of $\ell=1$. As in Corollary~\ref{cor:mmctime} (of \nameD) we can show that the latter is  $round_{\max} \in O(K^{4+\epsilon}\log^3 K)$, which is then the running time of each iteration~$K$.

Now we consider two cases. If $\lceil\lceil 12/\zeta\rceil\rceil$, which is the initial value of $K$, is larger than $\lceil\lceil 8n\rceil\rceil$, which is the value of $K$ in the last iteration, we know that \nameR will have only one iteration. Thus, the running time would be $O(K^{4+\epsilon}\log^3 K) \in O((1/\zeta)^{4+\epsilon}\log^3 (1/\zeta))$. 
Otherwise, if $\lceil\lceil 12/\zeta\rceil\rceil\leq\lceil\lceil 8n\rceil\rceil$, the number of rounds by the end of iteration
$K=\lceil\lceil 8n \rceil\rceil$ is
$O(n^{4+\epsilon}\log^3 n)$, since $K=O(n)$, by the end of iteration $K$
there were less than $\log K$ iterations, $K,K/2,K/2^2,\ldots$,
and each iteration $K/2^i$ had $O((K/2^i)^{4+\epsilon} \log^3(K/2^i))$ rounds.
Summing up these
times we get $O(K^{4+\epsilon}\log^3 K) \in O(n^{4+\epsilon}\log^3 n)$.
Thus, the claimed time to obtain the correct count follows.
\end{proof}

\section{Conclusions}
\label{sec:conclude}

This paper expanded the knowledge about feasibility of
polynomial computation in anonymous dynamic environment/networks.
In particular, a counting-type of computation could be done deterministically if symmetry
is broken by existing a partition of nodes in which one part
knows its size. It is also feasible without any distinction between the nodes with an arbitrary probability, albeit without termination.
Our algorithms are structured in phases, where nodes average some potential values. 
This phased process, which could be slower than a continuous one, is needed by our analysis. 
We leave the study of how to simplify the structure of the algorithm and to speed up the computation for future work.
Other natural open directions include studying randomized approximation
solutions and extension to other related models and problems,
as well as deriving tighter upper and lower bounds in the considered setting.

\section*{Acknowledgements}
We thank anonymous reviewers for numerous comments and suggestions to improve the quality of arguments and presentation.
This work was supported by 
the National Science Center Poland (NCN) grant 2017/25/B/ST6/02553;
the UK Royal Society International Exchanges 2017 Round 3 Grant \#170293; 
and Pace University SRC Grant and Kenan Fund.
%


\bibliography{Comprehensive_2010}

\clearpage


\end{document}